\documentclass[10pt,journal,compsoc]{IEEEtran}

\ifCLASSOPTIONcompsoc
  \usepackage[nocompress]{cite}
\else
  \usepackage{cite}
\fi
\ifCLASSINFOpdf
\else
\fi
\usepackage{soul}
\usepackage{amssymb}

\usepackage{algorithm}
\usepackage{algpseudocode}
\usepackage{pifont}
\usepackage{pgfplots}
\usepackage{xfrac}
\usetikzlibrary{arrows,calc}
\usepackage{blindtext}

\usepackage{amsmath}
\usepackage{multirow}

\usepackage{tikz}
\usepackage{siunitx}

\usepackage{url} 
\usepackage{amsmath,calc}
\usepackage{tikz}
\usepackage{pgfplots}
\usepackage{siunitx}
\usepackage{tikz}
\usepackage{pgfplots}
\pgfplotsset{width=7cm,compat=1.3}

\usepackage{psfrag}
\usepackage{tikz}
\usetikzlibrary{matrix}
\usepackage{pgfplots}
\usepackage{mathtools}
\usepackage{siunitx}
\usepackage[normalem]{ulem}
\usepackage{amsmath}
\newcommand{\stkout}[1]{\ifmmode\text{\sout{\ensuremath{#1}}}\else\sout{#1}\fi}

\usepackage{cite} 

\makeatletter
\let\OldStatex\Statex
\renewcommand{\Statex}[1][3]{%
	\setlength\@tempdima{\algorithmicindent}%
	\OldStatex\hskip\dimexpr#1\@tempdima\relax}
\makeatother

\newcommand{\shorten}[1]{}
\newtheorem{proposition}{Proposition}
\newtheorem{theorem}{Theorem}

\newcommand{\signed}%
{{\unskip\nobreak\hfill\penalty50
		\hskip2em\hbox{}\nobreak\hfil $\blacksquare$
		\parfillskip=0pt \finalhyphendemerits=0 \par}}
\newenvironment{proof}[1]
{
	\bf{Proof:}\rm{\noindent{#1 }}\ignorespaces
}
{\signed\addvspace\medskipamount}

\usepackage{scrextend}
\usepackage{inputenc}
\addtokomafont{labelinglabel}{\sffamily}

\begin{document}
%
\title{HashTag Erasure Codes: From Theory to Practice}
%
%
%
%

\author{Katina~Kralevska,
Danilo~Gligoroski,
Rune E.~Jensen,
and~Harald~{\O}verby
\IEEEcompsocitemizethanks{\IEEEcompsocthanksitem K. Kralevska, D. Gligoroski, and H. {\O}verby are with the Department of Information Security and Communication Technology, NTNU, Norwegian University of Science and Technology, Trondheim 7491, Norway. E-mail: \{katinak, danilog, haraldov\}@ntnu.no\protect

\IEEEcompsocthanksitem  R. E. Jensen is with the Department of Computer and Information Science, NTNU, Norwegian University of Science and Technology, Trondheim 7491, Norway. E-mail: runeerle@idi.ntnu.no \protect

}
}

\IEEEtitleabstractindextext{%
\begin{abstract}
Minimum-Storage Regenerating (MSR) codes have emerged as a viable alternative to Reed-Solomon (RS) codes as they minimize the repair bandwidth while they are still optimal in terms of reliability and storage overhead. 
Although several MSR constructions exist, so far they have not been practically implemented mainly due to the big number of I/O operations. 
In this paper, we analyze high-rate MDS codes that are simultaneously optimized in terms of storage, reliability, I/O operations, and repair-bandwidth for single and multiple failures of the systematic nodes. The codes were recently introduced in \cite{7463553} without any specific name. Due to the resemblance between the hashtag sign \# and the procedure of the code construction, we call them in this paper \emph{HashTag Erasure Codes (HTECs)}.
HTECs provide the lowest data-read and data-transfer, and thus the lowest repair time for an arbitrary sub-packetization level $\alpha$, where $\alpha \leq r^{\lceil \sfrac{k}{r} \rceil}$, among all existing MDS codes for distributed storage including MSR codes. The repair process is linear and highly parallel. Additionally, we show that HTECs are the first high-rate MDS codes that reduce the repair bandwidth for more than one failure. Practical implementations of HTECs in Hadoop release 3.0.0-alpha2 demonstrate their great potentials.
\end{abstract}

}

\maketitle

\IEEEdisplaynontitleabstractindextext

%
\IEEEpeerreviewmaketitle

\ifCLASSOPTIONcompsoc
\IEEEraisesectionheading{\section{Introduction} \label{intro}}
\else
\section{Introduction}
 \label{intro}
\fi
\IEEEPARstart{E}{rasure} coding has become a viable alternative to replication as it provides the same level of reliability with significantly less storage overhead \cite{Weatherspoon:2002:ECV:646334.687814}. 
When replication is used, the data is available as long as at least one copy still exists. The storage overhead of storing one extra replica is $100\%$, while it is $200\%$ for 2 replicas, and so forth. Therefore, replication is not suitable for large-scale storage systems. Its high storage overhead implies a high hardware cost (disk drives and associated equipment) as well as a high operational cost that includes building space, power, cooling, maintenance, etc.

Compared to replication, traditional erasure coding reduces the storage overhead but at a higher repair cost that is expressed through a high repair bandwidth (the amount of data transferred during the repair process), excessive input and output operations (I/Os), and expensive computations.
Practical implementations of erasure coding in distributed storage systems such as Google File System (GFS)\cite{Ghemawat:2003:GFS:945445.945450} and Hadoop Distributed File System (HDFS)\cite{Shvachko:2010:HDF:1913798.1914427} require Maximum Distance Separable (MDS), high-rate, and repair-efficient codes. 
Two primary metrics that determine the repair efficiency of a code are the amount of accessed data (\emph{data-read}) from the non-failed nodes and the amount of transferred data (\emph{repair bandwidth}).
We are interested in codes where these two metrics are \emph{minimized} and \emph{equal} at the same time because they are directly linked with the number of I/Os. The number of I/Os is an important parameter in storage systems especially for applications that serve a large number of user requests or perform data intensive computations where I/Os are becoming the primary bottleneck. There are two types of I/Os: sequential and random operations. \emph{Sequential operations} access locations on the storage device in a contiguous manner, while \emph{random operations} access locations in a non-contiguous manner.

Reed-Solomon (RS) codes \cite{10.2307/2098968} are a well-known representative of traditional \emph{MDS} codes.
Under a $(n, k)$ RS code, a file of $M$ symbols is stored across $n$ nodes with equal capacity of $\frac{M}{k}$ symbols. The missing/unavailable data from one node can be recovered from any $k$ out of $n$ nodes. Thus, a transfer of $k \times \frac{M}{k} = M$ symbols (the whole file) is needed in order to repair $\frac{1}{k}$-th of the file. A $(n, k)$ RS code tolerates up to $r=n-k$ failures without any permanent data-loss.
In general, the repair bandwidth and the number of I/Os are $k$ times higher with a $(n, k)$ RS code than with replication. 
The entire data from $k$ nodes has to be read during the recovery process with a RS code, hence the reads are sequential.

A powerful class of erasure codes that optimizes for repair bandwidth and storage costs has been proposed in \cite{5550492}. Minimum Bandwidth Regenerating (MBR) codes are optimal in terms of the repair bandwidth, while \emph{Minimum Storage Regenerating (MSR)} codes are optimal in terms of the storage.  
MSR codes possess all properties of MDS codes in addition to providing minimum repair bandwidth.
The repair bandwidth for a single failure for a $(n, k)$ MSR code is lower bounded by \cite{5550492}:
\begin{equation}
\gamma_{MSR}^{min} \leq \frac{M}{k} \frac{n-1}{n-k}.
\label{optimal}
\end{equation}
The length of the vector of symbols that a node stores in a single instance of the code determines the \emph{sub-packetization level} $\alpha$.
The data from a failed node is recovered by transferring $\beta$ symbols, where $\beta < \alpha$, from each of $d$ non-failed nodes called \emph{helpers}. Thus, the repair bandwidth $\gamma$ is equal to $d\beta$, where $\alpha \leq d\beta \ll M$.
The repair bandwidth for a single failure is minimized when a fraction of $\sfrac{1}{r}$-th of the stored data in all $d=n-1$ helpers is accessed and transferred, as it is shown in Eq. (\ref{optimal}). MDS codes that achieve the minimum repair bandwidth are called \emph{access-optimal codes}. However, the access operations are in a non-contiguous manner, meaning that the number of I/Os for MSR codes can be several orders of magnitude greater compared to the number of I/Os for RS codes. 

MSR codes optimize the repair bandwidth only for a single failure. 
Although single failures are dominant \cite{Rashmi:2014:HGF:2619239.2626325}, multiple failures are often correlated and co-occurring in practice \cite{179135,conf/fast/HuCLT12}. We believe that it is crucial to have MDS erasure codes that provide low repair bandwidth and low number of I/Os for any combination of $n, k,$ and $\alpha$ when recovering from a single and multiple failures in order to have a generally accepted practical implementation of erasure codes for distributed storage systems. 



\subsection{Our Contribution}

In this paper, we study both theoretical and practical aspects of the explicit construction of general sub-packetized erasure codes that was recently introduced in \cite{7463553}. Since the codes upon their definition were presented without any specific name, we call them here \emph{HashTag Erasure Codes (HTECs)}. 

The first contribution of this paper is that it provides many concrete instances of HTECs. HTEC construction is an explicit construction of MDS codes for an arbitrary sub-packetization level $\alpha$. We show that the bandwidth savings for a single failure with HTECs can be up to 60\% and 30\% compared to RS \cite{10.2307/2098968} and Piggyback codes \cite{6620242}, respectively. We also show that even for double failures, the codes still achieve bandwidth savings of 20\% compared to RS codes.
The code construction is general and works for any combination of $n, k,$ and $\alpha$ even when $r$ is not a divisor of $k$. 
HTECs are access-optimal codes for $\alpha=r^{\lceil\sfrac{k}{r}\rceil}$, i.e. they achieve the MSR point on the optimal trade-off curve between the storage and the repair bandwidth shown in Fig. \ref{curve}. 
For all other values of $\alpha$ that are less than $r^{\lceil\sfrac{k}{r}\rceil}$, HTECs achieve all points that lie from the MSR point to the conventional erasure code (EC) point on the curve in Fig. \ref{curve}.

\begin{figure}
	\centering
	\includegraphics[width=2.4in]{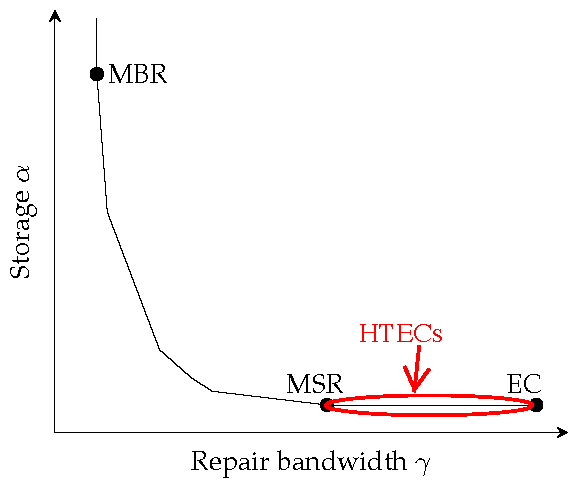}
	\vspace{-0.3cm}
	\caption{Regenerating codes (MSR and MBR) offer performance improvement compared to conventional erasure coding (EC). We propose explicit constructions of MDS codes that lie on the curve from the MSR point, including it, to the EC point.}\label{curve}
	\vspace{-0.4cm}
\end{figure}

The second contribution is that we elaborate the correlation between the repair bandwidth, the I/Os, and the repair time in terms of $\alpha$. 
While large values of $\alpha$ guarantee low average repair bandwidth, they inevitably increase the number of I/Os (in particular the number of random access operations) that has an impact on the repair time, the throughput, the CPU utilization, and the data availability. 
Hadoop measurements show that HTECs can be of a great practical importance in distributed storage systems as they provide the system designer with greater flexibility in terms of selecting various code parameters such as the rate of the code, the size of the blocks and splits of the files, and various values of $\alpha$ in order to fine tune and minimize the overall repair time. 
We show that all $k$ systematic nodes are clustered in subsets of $r$ nodes (with the exception of the last subset that can have less than $r$ nodes). Then, we prove that there is one subset that can be repaired with $n-1$ sequential reads. Further on, for all other subsets of $r$ nodes the discontiguity increases sequentially. 

Our third contribution is a deeper scrutiny of the repair process with HTECs. In general, the repair process for a single failure is linear and highly parallel. This means a set of $\lceil \sfrac{\alpha}{r}\rceil$ symbols is independently repaired first and used along with the accessed data from other nodes to repair the remaining symbols of the failed node. 
We show that HTECs have one extra beneficial feature compared to RS codes: the amount of accessed and transferred data when multiple failures occur is less than RS codes. 
To the best of the authors' knowledge, HTECs are the first codes in the literature that offer bandwidth savings when recovering from multiple failures for any code parameters including the high-rate regime. 

\begin{table*}[t]
	\caption{Comparison of Explicit MDS Code Constructions}\label{compare}
	\vspace{-0.2cm}
	\centering
	\begin{tabular}{|c|c|c|c|c|c|}
		\hline	\rule{0pt}{3ex} 		
		Code construction & MDS & $k, r$ parameters & Sub-packetization level & Optimized for $t$ failures \\
		\hline \rule{0pt}{3ex} 
		High-rate MSR \cite{7084873} & Y & $r=\sfrac{k}{m}, m\geq 1$ & $r^\frac{k}{r}$ & $t=1$\\
		\hline \rule{0pt}{3ex} 
		MSR over small fields \cite{DBLP:journals/corr/RavivSE15} & Y & $r=2,3$ & $r^\frac{k}{r}$ & $t=1$\\
		\hline \rule{0pt}{3ex}
		Product-Matrix MSR \cite{5961826} & Y & $r\geq k-1$ & $r$ & $t=1$\\
		\hline  \rule{0pt}{3ex}
		Piggyback 1 \cite{6620242} & Y & All & $2 m$, $m\geq 1$ & $t=1$\\
		\hline \rule{0pt}{3ex}
		Piggyback 2 \cite{6620242} & Y & $r\geq 3$ & $(2r-3)m$, $m\geq 1$ & $t=1$\\
		\hline \rule{0pt}{3ex}
		Rotated RS \cite{conf/fast/KhanBPPH12} & Y & $r\in\{2,3\},k\leq36$ & 2 & $t=1$\\
		\hline \rule{0pt}{3ex}
		EVENODD, RDP \cite{364531,corbett:rdp} & Y & $r=2$ & $k$ & $t=1$\\
		\hline \rule{0pt}{3ex}
		MSCR \cite{conf/isit/ChenS13} & Y & $r=k$ & $r$ & $2\leq t \leq r$\\
		\hline \rule{0pt}{3ex}
		CORE \cite{6880379} & Y & $r=k$ & $r$ & $1\leq t \leq r$\\
		\hline \rule{0pt}{3ex}
		New Piggyback \cite{DBLP:journals/corr/ShangguanG16a} & Y & $r\ll k$ & $r$ & $t=1$\\
		\hline \rule{0pt}{3ex}
		HashTag Erasure Codes (HTEC) & Y & All & $2\leq \alpha\leq r^{\lceil\sfrac{k}{r}\rceil}$ & $1\leq t \leq r$\\
		\hline 
	\end{tabular}
	\vspace{-0.2cm}
\end{table*}

\subsection{Paper Outline}
The rest of the paper is organized as follows. Section \ref{RW} reviews the state-of-the-art for MDS erasure codes for distributed storage.
Section \ref{mathPrel} provides the mathematical preliminaries and properties of HTECs. HTECs examples are presented in Section \ref{examples}.
An algorithm for I/O optimization is presented in Section \ref{ioOpt}.
Hadoop measurements and performance comparisons between HTECs and representative codes from the literature are given in Section \ref{anal}. 
Section \ref{dis} discusses open issues, and Section \ref{conc} concludes the paper.

\section{Related Works}\label{RW}

There has been a considerable amount of work in the area of erasure codes for distributed storage. 
We only review the most relevant literature about \emph{exact repair codes} where the reconstructed data is exactly the same as the lost data because HTECs belong to the class of exact repair codes. 
Table \ref{compare} compares a selection of codes with respect to the MDS property, the supported parameters, the sub-packetization level, and the number of failures they are optimized for.

In \cite{6416066}, high-rate $(n, k=n-2, d=n-1)$ MSR codes using Hadamard designs for optimal repair of both systematic and parity nodes were constructed. 
The work also presented a general construction for optimal repair only of the systematic nodes for $\alpha$ of $r^k$.
Codes for optimal systematic-repair for the same $\alpha$ appeared in \cite{6033730} and \cite{6352912}.
The work was subsequently extended in \cite{6120327} to include repair of the parity nodes.

Furthermore, the work in \cite{6283042} showed that the required $\alpha$ for construction of access-optimal MSR codes is $r^\frac{k}{r}$. 
Few code constructions for optimal systematic-repair for $\alpha=r^\frac{k}{r}$ followed in the literature.
In \cite{6190343}, Cadambe et al. proposed a high-rate MSR construction that is not explicit and requires a large field size. 
Later, an alternate construction of access-optimal MSR codes for $\alpha=r^m$, where $m=\frac{k}{r}$, was presented in \cite{7084873}.
An essential condition for the code construction in \cite{7084873} is that $m$ has to be an integer $m \geq 1$, where $k$ is set to $rm$ and $\alpha$ to $r^m$. 
Explicit access-optimal systematic-repair MSR codes over small finite fields for $\alpha \geq r^{\frac{k}{r}}$ were presented in \cite{DBLP:journals/corr/RavivSE15}. However, these codes exist only for $r=2,3$.

Although the aforementioned constructions achieve the lower bound of the repair bandwidth for a single failure, they have not been practically implemented in real-world distributed storage systems. There are at least two reasons for that: either MSR codes require encoding/decoding operations over an exponentially growing finite field or $\alpha$ increases exponentially.  
Practical implementations of erasure coding \cite{Rashmi:2014:HGF:2619239.2626325,194416} showed that a good erasure code has to provide a satisfactory tradeoff between the system-level metrics such as storage overhead, reliability, repair bandwidth, and I/Os. One way to achieve a satisfactory system tradeoff is to work with small sub-packetization levels.

Piggyback codes \cite{6620242} are a good example of practical MDS codes with small $\alpha$. The basic idea of the piggyback framework is to take multiple instances of an existing code and add carefully designed functions of the data from one instance to another.
Piggyback codes have better repair bandwidth performance than Rotated-RS \cite{conf/fast/KhanBPPH12}, EVENODD \cite{364531}, and RDP codes \cite{corbett:rdp}. Rotated-RS codes exist for $r \in \{2,3\}$ and $k\leq 36$, while EVENODD and RDP exist only for $r=2$.
The idea of the piggyback framework has been adopted in several works \cite{Rashmi:2014:HGF:2619239.2626325,7463553,DBLP:journals/corr/ShangguanG16a,7222447}.
In \cite{Rashmi:2014:HGF:2619239.2626325}, Rashmi et al. reported bandwidth savings of $35\%$ for a $(14, 10)$ code with $\alpha=2$ when repairing a systematic node compared to a $(14, 10)$ RS code.
A $(14, 10)$ HTEC \cite{7463553}, studied also in this paper, offers bandwidth savings of 41\% for $\alpha=2$, and the bandwidth savings can go up to 67.5\% for $\alpha$ equal to $r^{\lceil \sfrac{k}{r} \rceil}=64$. 
A flexible code construction such as the one in \cite{7463553} has not been presented in \cite{7084873,6190343,6283042,DBLP:journals/corr/ShangguanG16a,7222447}.
A new piggyback design that achieves a repair bandwidth of $\frac{\sqrt{2r-1}}{r}$ for a systematic node repair was recently presented in \cite{DBLP:journals/corr/ShangguanG16a}. The limitation of the new piggyback design is that $\alpha$ is equal to $r$ and $r\ll k$.
Yang et al. \cite{7222447} applied piggybacking to optimize the repair of parity nodes while retaining the optimal repair bandwidth for systematic nodes. 

Another way to improve the system performance is I/O optimization while still keeping the storage and bandwidth optimality.
An algorithm that transforms Product-Matrix-MSR codes \cite{5961826} into I/O optimal codes (termed PM-RBT codes) was presented in \cite{188424}. However, PM-RBT codes exist only for the low-rate regime, i.e. $r\geq k-1$. 

All MDS erasure codes discussed previously optimize the repair of a single node failure. 
A cooperative recovery mechanism in the minimum-storage regime for repairing from multiple failures was proposed in \cite{5402494,5502493}. 
Minimum Storage Collaborative Regenerating (MSCR) codes minimize the repair bandwidth while still keeping the MDS property by allowing new nodes to download data from the non-failed nodes and the new nodes to exchange data among themselves.
The repair bandwidth for MSCR codes under functional repair was independently derived in \cite{5402494} and \cite{5978920}.
The authors in \cite{6847953} showed that it is possible to construct exact MSCR codes for optimal repair of two failures directly from existing exact MSR codes.
MSCR codes that cooperatively repair any number of systematic nodes, parity nodes, or a combination of one systematic and one parity node were presented in \cite{conf/isit/ChenS13}. However, their code rate is low ($n=2k$). 
A study about the practical aspects of codes for the same rate ($n=2k$) in a system called CORE that supports multiple failures can be found in \cite{6880379}.
There is no explicit construction of high-rate MDS codes for exact repair from multiple failures at the time of writing of this paper.

\section{Mathematical preliminaries}\label{mathPrel}
We consider systematic coding where $k$ nodes store the original data without encoding it. We refer to these $k$ nodes as systematic nodes and the remaining $r=n-k$ nodes are called parity nodes. 
Additionally, the codes are MDS. A $(n, k)$ MDS code is optimal in terms of the \emph{storage-reliability tradeoff} because it offers maximum fault tolerance of up to $r$ arbitrary failures for the added storage overhead of $r$ nodes.

Dimakis et al. introduced the repair bandwidth as a new metric for repair efficiency of erasure codes \cite{5550492}. MDS codes that achieve the lower bound of the repair bandwidth given in (1) are optimal with respect to the \emph{storage-bandwidth tradeoff}.
MSR codes are optimal with respect to both storage-reliability and storage-bandwidth tradeoffs. The following table summarizes the notation used in this paper.
\begin{table}[htb]
	\vspace{-0.2cm}
	\begin{center}
		\begin{tabular}{|l|p{70mm}|}
			\hline
			$n$ & Total number of nodes \\
			$k$ & Number of systematic nodes \\
			$r$ & Number of parity nodes \\
			$d$ & Number of non-failed nodes (helpers) \\ 
			$t$ & Number of failed nodes (failures), $1\leq t\leq r$ \\ 
			$d_j$ & The $j$-th systematic node, $1\leq j\leq k$ \\
			$p_l$ & The $l$-th parity node, $1\leq l\leq r$ \\
			$a_{i,j}$ & The $i$-th element of the $j$-th systematic node, $1\leq i\leq \alpha$ and $1\leq j\leq k$ \\
			$p_{i,l}$ & The $i$-th element of the $l$-th parity node, $1\leq i\leq \alpha$ and $1\leq l\leq r$ \\
			$\mathbf{F}_q$ & Finite field of size $q$ \\
			$q$ & Size $q$ of a finite field \\ 
			$c_{l,i,j}$ & Non-zero coefficient from a finite field, $1\leq l\leq r$, $1\leq i\leq \alpha$ and $1\leq j\leq k$ \\
			$M$ & Size of the original data \\
			$\alpha$ & Sub-packetization level \\
			$\beta$ & Amount of transferred data from a node \\
			$\gamma$ & Total amount of accessed and transferred data per node repair \\
			\hline
		\end{tabular}
		\vspace{-0.4cm}
	\end{center}
\end{table}


We have already presented some general properties of RS and MSR codes in Section \ref{intro}.
We illustrate them via examples in the next two Subsections where we also motivate for the need of HTECs.

\subsection{An Example with Reed-Solomon Codes}
Let us consider the following example with a RS code for $k=6$ and $r=3$.
The storage overhead is $50\%$ and the code can recover from up to 3 failures.
Fig. \ref{stor}a depicts the storage of a file of 54MB across 9 nodes where each node stores 9MB. It also illustrates the reconstruction of the node $d_1$ from the nodes $d_2, \ldots, d_6$ and $p_1$. 
In order to reconstruct 9MB of unavailable data, $6\times9$MB$=$54MB are read from 6 nodes and transferred across the network to the node performing the decoding computations.
The same amount of data (54MB) is needed to repair from 1, 2, and 3 node failures.
The number of random reads for this example is 6 since data is read in a contiguous manner from 6 different locations.

\subsection{Two Examples with HashTag Erasure Codes}
First we illustrate the performance improvement with an $(9, 6)$ access-optimal HTEC for $\alpha=9$ and $d=8$ compared to a $(9, 6)$ RS code.
The bandwidth to repair any systematic node is reduced by 55.6\% compared to a $(9, 6)$ RS code. 
The reconstruction of node $d_1$ is illustrated in Fig. \ref{stor}b. 
In order to reconstruct the data from $d_1$, $1/3$-rd of the stored data from all 8 helpers is accessed and transferred, hence the repair bandwidth is only 24MB compared to 54MB with RS. 
HTECs achieve the minimum repair bandwidth given in Eq. (\ref{optimal}) when $\alpha=9$.
However, contacting 8 nodes when recovering from a single failure increases the number of seek operations and random I/Os.

A smaller value of $\alpha$ reduces the number of random I/Os. The repair of $d_1$ presented in Fig. \ref{stor}c is for $\alpha=6$. Namely, 3MB are accessed and transferred from $d_2, d_3, d_6, p_1, p_2,$ and $p_3$, while 4.5MB from $d_4$ and $d_5$. Thus, the repair bandwidth for $d_1$ is 27MB. The same amount of data is needed to repair $d_3, d_4,$ and $d_6$, while 30MB of data is needed to repair $d_2$ and $d_5$. The average repair bandwidth, defined as the ratio of the total repair bandwidth to repair all systematic nodes to the file size, for the systematic nodes is 28MB. Implementing an erasure code for $\alpha=6$ is simpler than an erasure code for $\alpha=9$, and it still provides bandwidth savings compared to 54MB with RS while it is slightly more than 24MB with the MSR code for $\alpha=9$. 
The big savings that come from the bandwidth reduction are evident when storing petabytes of data.

\begin{figure}
	\centering
	\includegraphics[width=3.45in]{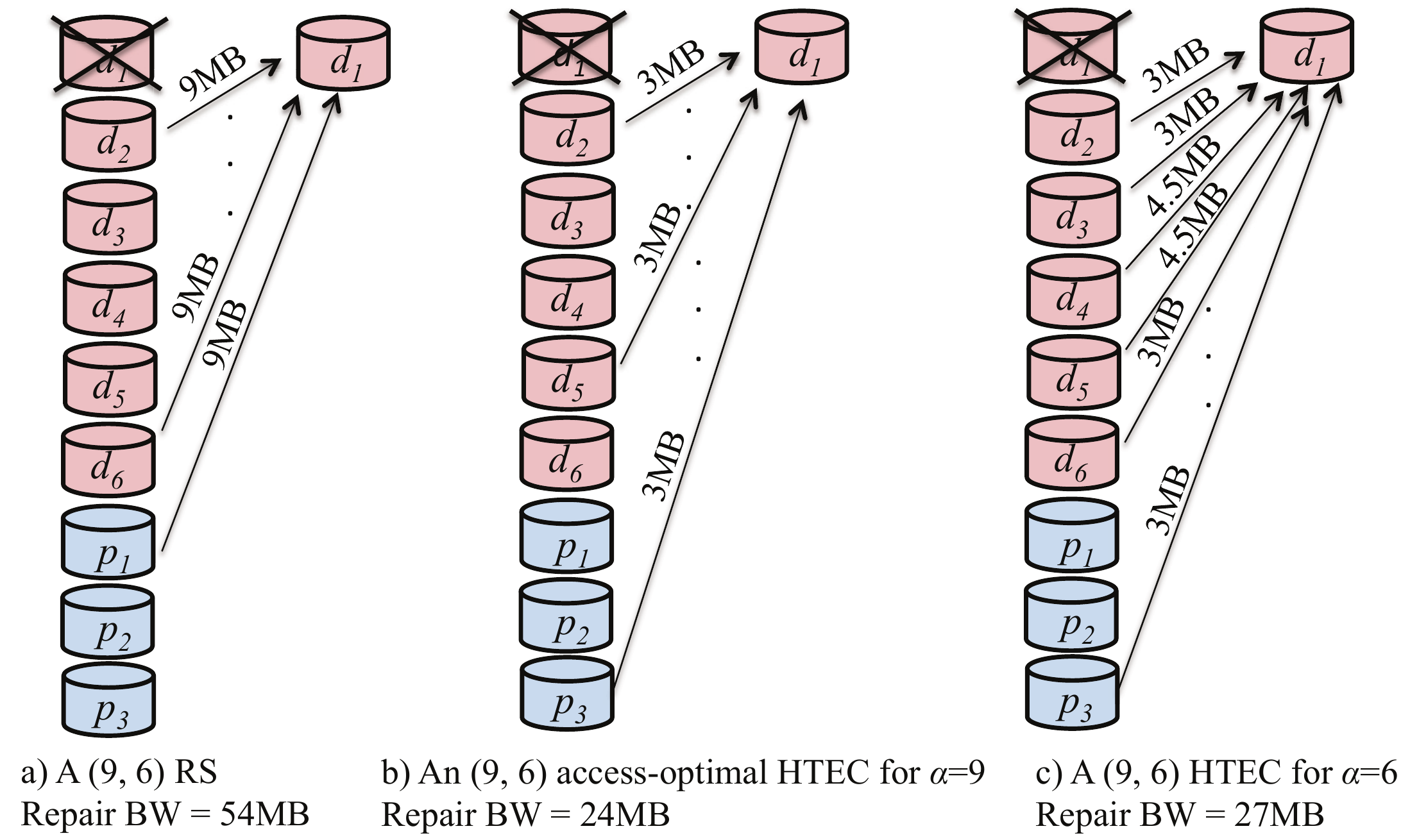}
	\vspace{-0.3cm}
	\caption{Amount of accessed and transferred data for repair of the systematic node $d_1$ for a (9,6) RS code, an (9, 6) access-optimal HTEC for $\alpha=9$, and a (9, 6) HTEC for $\alpha=6$. The systematic nodes are represented in red and the parity nodes in blue.}
	\label{stor}
	\vspace{-0.2cm}
\end{figure}





\subsection{Definition of HashTag Erasure Codes} \label{algorithm}
Consider a file of size $M = k \alpha$ symbols from a finite field $\mathbf{F}_q$ stored in $k$ systematic nodes $d_j$ of capacity $\alpha$ symbols.
The general algorithm introduced in \cite{7463553} offers a rich design space for constructing HTECs for various combinations of $k$ systematic nodes, $r$ parity nodes (the total number of nodes is $n=k+r$), and sub-packetization levels $\alpha$. 

As a general notation we say that a systematic node $d_j$, where $1\le j \le k$, consists of an indexed set of $\alpha$ symbols  $\{a_{1,j},a_{2,j},\ldots,a_{\alpha,j}\}$.
The set $N=\{d_1, \ldots, d_k\}$ of $k$ systematic nodes is partitioned in $\lceil \sfrac{k}{r}\rceil$ disjunctive subsets $J_1, J_2, \ldots, J_{\lceil \sfrac{k}{r}\rceil}$ where $|J_\nu|=r$ (if $r$ does not divide $k$ then $J_{\lceil \sfrac{k}{r}\rceil}$ has $k \mod{r}$ elements) and $N=\cup_{\nu=1}^{\lceil \sfrac{k}{r}\rceil}J_\nu$. In general, the partitioning can be any random permutation of $k$ nodes. Without loss of generality we use the natural ordering as follows:
$J_1=\{d_1,\ldots,d_r \} $, $J_2=\{d_{r+1},\ldots,d_{2r} \} $, $\ldots$ , $J_{\lceil \sfrac{k}{r}\rceil} = \{d_{ \lfloor \sfrac{k}{r} \rfloor \times r + 1},\ldots,d_{k} \} $. 

	\begin{center}
	\includegraphics[width=3.3in]{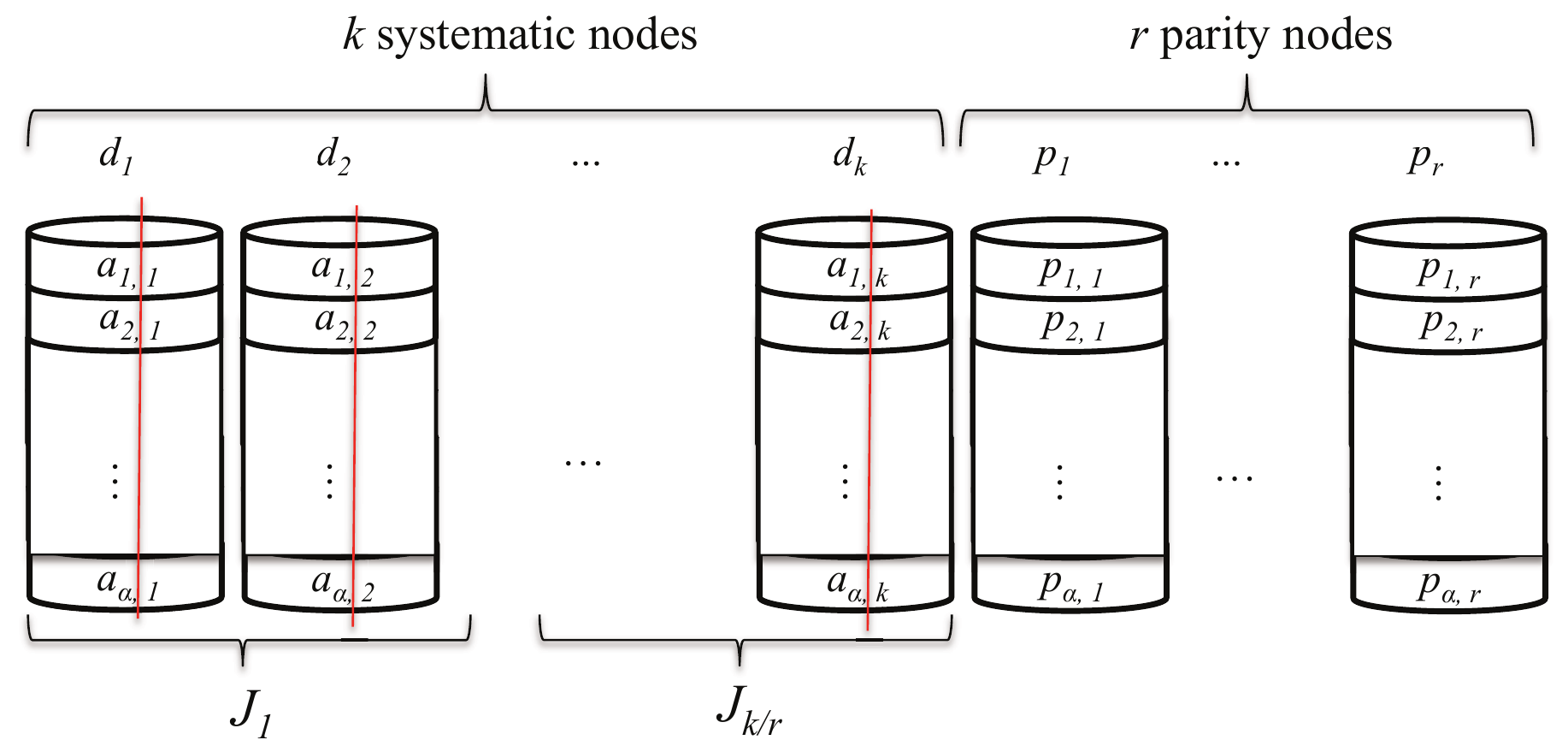}
	\vspace{-0.4cm}
	\label{general}
\end{center}
The basic idea for generating the linear dependencies for the parity nodes $p_l, l=1,\ldots,r,$ can be described as setting up $r$ grids where the pair of indexes $(i, j)$ of the symbols $a_{i, j}$ of each systematic node (struck through with red lines) are represented as columns in a grid (resemblance to the vertical lines in the hashtag sign \# in the expression for $P_i$), and where the linear dependencies for the symbols $p_{i,l}, i=1,\ldots,\alpha$ and $l=1,\ldots,r$, of the parity nodes are obtained as linear combinations from the elements which indexes are represented in the rows of the grid (resemblance to the horizontal lines in \#). Consequently, the name \emph{HashTag Erasure Codes (HTECs)} comes from the resemblance between the code construction and the hashtag sign \#. 

In other words, the basic data structure component for construction of HTECs is an index array $P = ((i,j))_{\alpha \times k}$ of size $\alpha \times k$, where $\alpha \leq r^{\lceil \sfrac{k}{r}\rceil}$. The index arrays are generated by Alg. \ref{HighLevel}, Alg. \ref{Encode}, and Alg. \ref{Valid}. Alg. \ref{HighLevel} is a high level algorithm, while Alg. \ref{Encode} is a detailed algorithm that calls Alg. \ref{Valid} for splitting the symbols following defined conditions. 
In the initialization phase of Step 1 in Alg. \ref{Encode}, $r$ index arrays are constructed as follows:

\begin{tikzpicture}
\hspace{1.0cm}\matrix (mat) [%
matrix of nodes,
left delimiter={[},right delimiter={]}
]
{%
	$(1, 1)$ & $(1, 2)$ & \ldots & $(1, k)$\\
	$(2, 1)$ & $(2, 2)$ & \ldots & $(2, k)$\\
	$\vdots$ & $\vdots$ & $\ddots$ & $\vdots$\\
	$(\alpha, 1)$ & $(\alpha, 2)$ & \ldots & $(\alpha, k)$\\
};
\draw[black] (mat-1-1.west)  -- (mat-1-4.east);
\draw[black] (mat-2-1.west)  -- (mat-2-4.east);
\draw[red] (mat-1-1.north) -- (mat-4-1.south);
\draw[red] (mat-1-2.north) -- (mat-4-2.south);
\draw[red] (mat-1-4.north) -- (mat-4-4.south);
\node [left] at (-2.5,0) {$P_i = $};
\end{tikzpicture}
	
In Step 2, additional $\lceil \sfrac{k}{r}\rceil$ columns with pairs $(0, 0)$ are added to $P_2, \ldots, P_r$ as:
$$\ \ \ \ \ \ \ \ \ \ \ \ \ \ \ \ \ \ \ \ \ \ \ \ \ \ \ \ \ \ \ \ \ \ \ \ \ \ \ \ \ \ \ \ \ \ \ \ \ \ \ \ \ \overbrace{ \ \ \ \ \ \ \ \ \ \ \ \ \ \ \ \ \ \ \ \ \ \ \ \ \ \ \ \ \ }^{\lceil \frac{k}{r} \rceil}$$
\begin{equation*}
P_i=
\begin{bmatrix}
(1, 1) & (1, 2) & \ldots & (1, k) &  (0, 0) & \ldots & (0, 0) \\
(2, 1) & (2, 2) & \ldots & (2, k) & (0, 0) & \ldots & (0, 0) \\
\vdotswithin{1} & \vdotswithin{\alpha_n} & \ddots & \vdotswithin{{\alpha_n}^{k-1}}\\
(\alpha, 1) & (\alpha, 2) & \ldots & (\alpha, k) & (0, 0) & \ldots & (0, 0) \\
\end{bmatrix}.
\end{equation*}


In the next steps of Alg. \ref{Encode}, the zero pairs are replaced with concrete $(i, j)$ pairs so that the repair bandwidth is minimized for a given sub-packetization level $\alpha$. The set of all symbols in $d_j$ is partitioned in disjunctive subsets where at least one subset has $portion = \lceil \sfrac{\alpha}{r}\rceil$ number of elements. The values of $\alpha$ and $k$ determine two phases of the algorithm. 
The first phase starts with a granulation level parameter called $run$ that is initialized to $\lceil \sfrac{\alpha}{r}\rceil$ and a parameter called $step$ initialized to 0. 
These parameters affect how the $i$ indexes of the elements in the systematic nodes are scheduled. 
The set of indexes $D=\{1,\ldots,\alpha\}$, where the $i-$th index of $a_{i,j}$ from $d_j$ is represented by $i$ in $D$, is partitioned in $r$ disjunctive subsets $D = \cup_{\rho=1}^{r}D_{\rho,d_j}$ where each subset has $portion$ elements.
If the elements in the subsets of the partition for node $d_j$, $\mathcal{D}_{d_j} = \{D_{1,d_j},\ldots,D_{r,d_j}\}$, are taken in runs of $run$ consecutive elements with distance equal to $step$, then the partition is called a \emph{valid partition}. 
For the subsequent $(k+\nu)$-th column in $P_2, \ldots, P_r$, where $\nu \in \{1,\ldots,\lceil \sfrac{k}{r}\rceil\}$, the scheduling of the indexes corresponding to the elements from the nodes in $J_\nu$ is done in subsets of indexes from a \emph{valid partition} that is an output of Alg. \ref{Valid}. If $r$ divides $\alpha$, then the \emph{valid partition} for all nodes in $J_\nu$ is equal. If $r$ does not divide $\alpha$, then the \emph{valid partition} has to contain at least one subset $D_{\rho,d_j}$ with $portion$ elements that correspond to the row indexes in the $(k+\nu)$-th column in one of the arrays $P_2, \ldots, P_r$ that are all zero pairs. 

\emph{\textbf{Example 1.}}
	Let us take a $(9, 6)$ code for $\alpha = 9$. The systematic nodes are split into two subsets, $J_1=\{d_1, d_2, d_3\}$ and $J_2=\{d_4, d_5, d_6\}$ as presented in Fig. \ref{9_6_9}.
	The set of the elements corresponding to the $i$ indexes of $a_{i,1}$ from $d_1$ is represented as $D=\{1, 2, \ldots, 9\}$. Since $portion = \lceil \sfrac{9}{3}\rceil=3$, then $D$ is divided into subsets of 3 elements. Following Alg. \ref{Encode}, $run=3$ and $step=0$ for all nodes in $J_1$. Additionally, $r=3$ divides $\alpha=9$, thus, $\mathcal{D}_{d_1}=\mathcal{D}_{d_2}=\mathcal{D}_{d_3}$. The \emph{valid partition} $\mathcal{D}_{d_1}$ that is an output from Alg. \ref{Valid} is obtained as follows:
	 	\begin{center}
	 	\includegraphics[width=2.5in]{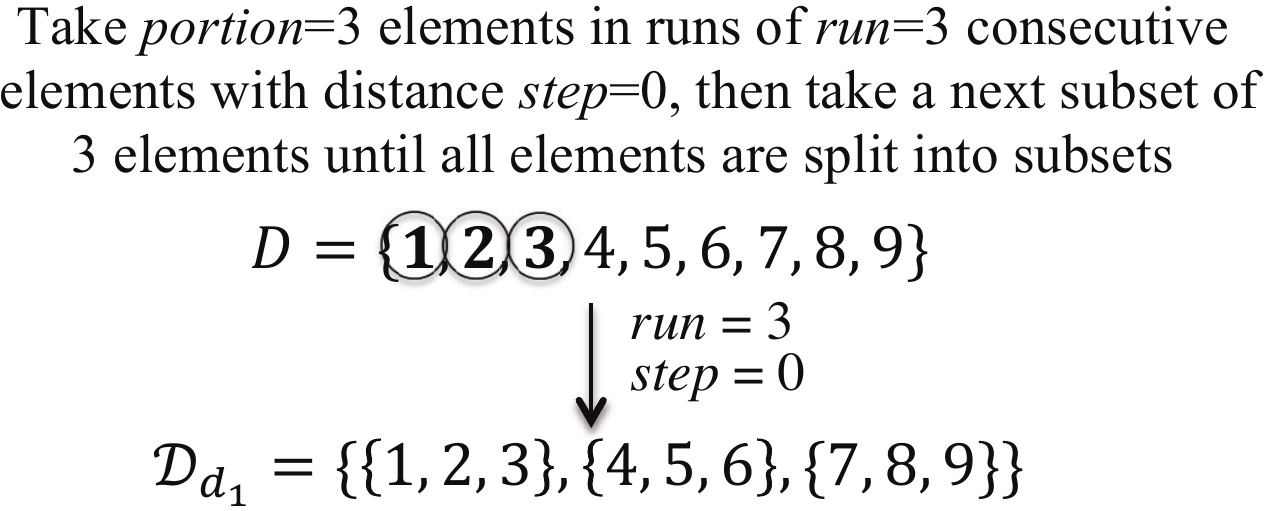}
	 	\label{partitioning}
	 \end{center}
	When calculating the \emph{valid partition} for the nodes in $J_2$, the same steps are performed but for $run=1$ and $step=2$. Consequently, the \emph{valid partition} $\mathcal{D}_{d_4}$ that is equal to $\mathcal{D}_{d_5}$ and $\mathcal{D}_{d_6}$ is generated as follows:
		\begin{center}
		\includegraphics[width=2.4in]{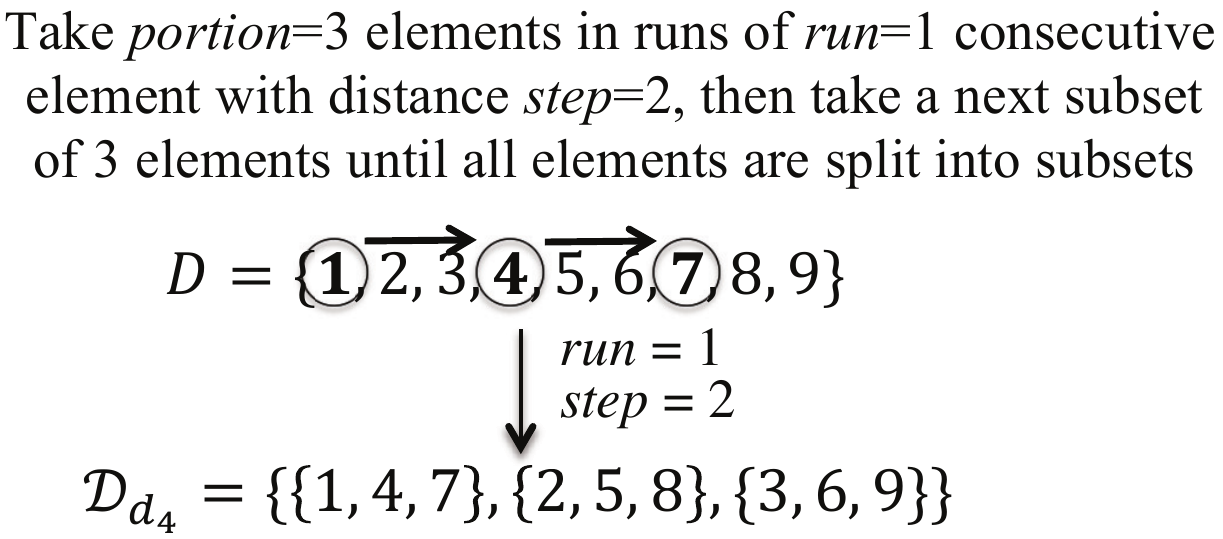}
			\label{partitioning}
		\end{center}
	The presented explanation so far corresponds to Steps 1-13 from Alg. \ref{Encode}. In Steps 14 and 15, the zero pairs in the index arrays $P_i$ (i.e. $P_2$ and $P_3$) are filled using the so far produced information about \emph{ValidPartitions}.
	The first 3 zero pairs in the $7$-th column of $P_2$ with $0$ distance between them are at the positions $(1,7), (2, 7)$ and $(3,7)$, thus, we choose the subset $\{1,2,3\}$ out of $\mathcal{D}_{d_1}$ to be $D_{\rho,d_1}$.
 	The indexes of the elements of $d_1$ with $i$ indexes that are not elements of $D_{\rho,d_1}$ (represented in red color in Fig. \ref{9_6_9}) are scheduled in the 1-st, 2-nd, and 3-rd row and 7-th column of $P_2$ and $P_3$.
	Similarly, we perform the same steps for all systematic nodes, and the corresponding $D_{\rho,d_j}$, where $j=1, \ldots, 6,$ are given as subsets of $D_1$ in Table \ref{Comparison3Partitions}. The final scheduling of the elements is presented in Fig. \ref{9_6_9}.
	A more detailed explanation for a $(9, 6)$ code for $\alpha=6$ can be found in Section \ref{explanation}.\ \ $\blacksquare$
	
	In the first phase, the $(i,j)$ pairs that replace the $(0, 0)$ pairs are chosen such that both Condition 1 and Condition 2 are satisfied. The granulation level $run$ decreases by a factor $r$ with every round. 
Once $run$ becomes equal to 1 and there are still $(0, 0)$ pairs that have to get some $(i, j)$ values from the unscheduled elements in the systematic nodes, the second phase starts where the remaining indexes are chosen such that only Condition 2 is satisfied. 
\begin{itemize}
	\item \emph{Condition 1}: 
	At least one subset $D_{\rho,d_j}$ has $portion=\lceil \sfrac{\alpha}{r}\rceil$ elements with runs of $run$ consecutive elements separated with a distance between the indexes equal to $step$. The elements of that subset correspond to the row indexes in the $(k+\nu)$-th column, where $\nu=1,\ldots,\lceil \sfrac{k}{r}\rceil,$ in one of the arrays $P_2, \ldots, P_r$ that are all zero pairs. The distance between two elements in one node is computed in a cyclical manner such that the distance between the elements $a_{\alpha-1}$ and $a_2$ is 2.
	
	
	
	\item \emph{Condition 2}: A necessary condition for the valid partition to achieve the lowest possible repair bandwidth is $\mathcal{D}_{d_{j_1}} = \mathcal{D}_{d_{j_2}}$ for all $d_{j_1}$ and $d_{j_2}$ in $J_\nu$ and $D_{\rho,d_{j_1}}\neq D_{\rho,d_{j_2}}$ for all $d_{j_1}$ and $d_{j_2}$ systematic nodes in the system.
	If $portion=\lceil \sfrac{\alpha}{r}\rceil$ divides $\alpha$, then $D_{\rho,d_j}$ for all $d_j$ in $J_\nu$ are disjunctive, i.e. $D = \cup_{j=1}^{r}D_{\rho,d_j}=\{1, \ldots, \alpha\}$.
\end{itemize}

\begin{algorithm}
	\small
	\caption{High level description of an algorithm for generating HTEC for an arbitrary sub-packetization level
		\newline
		\textbf{Input:} $n, k, \alpha$;
		\newline
		\textbf{Output:} Index arrays $P_1, \ldots, P_r$.}
	\label{HighLevel}
	\begin{algorithmic}[1]
		\State{\textbf{Initialization:} $P_1, \ldots, P_r$ are initialized as index arrays $P = ((i,j))_{\alpha \times k}$;}
		\State{Append $\lceil \sfrac{k}{r}\rceil$ columns to $P_2, \ldots, P_r$ all initialized to $(0, 0)$;}
		\State \# Phase 1
		\State Set the granulation level $run \leftarrow \lceil \sfrac{\alpha}{r}\rceil$;
		\Repeat
		\State \parbox[t]{8cm}{Replace $(0, 0)$ pairs with indexes $(i, j)$ such that both Condition 1 and Condition 2 are satisfied;  \vspace{0.1cm}}
		\State \parbox[t]{8cm}{Decrease the granulation level $run$ by a factor $r$;}
		\Until{The granulation level $run > 1$}
		\State \# Phase 2
		\State If there are still $(0, 0)$ and unscheduled elements from the systematic nodes, choose $(i, j)$ such that only Condition 2 is satisfied;
		\State Return the index arrays $P_1, \ldots, P_r$.
	\end{algorithmic}
\end{algorithm}

\begin{algorithm}
	\small
	\caption{Algorithm to generate the index arrays
		\newline
		\textbf{Input:} $n, k, \alpha$;
		\newline
		\textbf{Output:} Index arrays $P_1, \ldots, P_r$.}
	\label{Encode}
	\begin{algorithmic}[1]
		\State{\textbf{Initialization:} $P_1, \ldots, P_r$ are initialized as index arrays $P = ((i,j))_{\alpha \times k}$;}
		\State{Append $\lceil \sfrac{k}{r}\rceil$ columns to $P_2, \ldots, P_r$ all initialized to $(0, 0)$;}
		\State Set $portion \leftarrow \lceil \sfrac{\alpha}{r}\rceil$;
		\State Set $ValidPartitions \leftarrow \emptyset$;
		\State Set $j \leftarrow 0$;
		\State \# Phase 1
		\Repeat
		\State Set $j \leftarrow j+1$;
		\State Set $\nu \leftarrow \lceil \sfrac{j}{r}\rceil$;
		\State Set $run \leftarrow \lceil \sfrac{\alpha}{r^\nu}\rceil$;
		\State Set $step \leftarrow \lceil \sfrac{\alpha}{r}\rceil-run$;
		\State \parbox[t]{8cm}{$\mathcal{D}_{d_j} =$ $ValidPartitioning(ValidPartitions$, $k$, $r$, $portion$, $run$, $step$, $J_{\nu})$;  \vspace{0.1cm}}
		\State \parbox[t]{8cm}{Set $ValidPartitions=ValidPartitions \cup \mathcal{D}_{d_j}$; \vspace{0.1cm}}
		\State \parbox[t]{8cm}{Determine one $D_{\rho,d_j} \in \mathcal{D}_{d_j}$ such that its elements correspond to row indexes in the $(k+\nu)$-th column in one of the arrays $P_2, \ldots, P_r$, that are all zero pairs $(0, 0)$; \vspace{0.1cm}}
		\State \parbox[t]{8cm}{The indexes in $D_{\rho,d_j}$ are the row positions where the pairs $(i,j)$ with indexes $i \in \mathcal{D} \setminus D_{\rho,d_j}$ are assigned in the $(k+\nu)$-th column of $P_2, \ldots, P_r$; \vspace{0.1cm}}
		\Until{$(run>1)\ \ \mbox{AND} \ \ (j \ne 0 \mod{r}) $}
		\State \# Phase 2
		\While{$j < k$}
		\State Set $j \leftarrow j+1$;
		\State Set $\nu \leftarrow \lceil \sfrac{j}{r}\rceil$;
		\State Set $run \leftarrow 0$;
		\State \parbox[t]{8cm}{$\mathcal{D}_{d_j} =$ $ValidPartitioning(ValidPartitions$, $k$, $r$, $portion$, $run$, $step$, $J_{\nu})$;  \vspace{0.1cm}}
		\State \parbox[t]{8cm}{Set $ValidPartitions=ValidPartitions \cup \mathcal{D}_{d_j}$; \vspace{0.1cm}}
		\State \parbox[t]{8cm}{Determine one $D_{\rho,d_j} \in \mathcal{D}_{d_j}$ such that its elements correspond to row indexes in the $(k+\nu)$-th column in one of the arrays $P_2, \ldots, P_r$, that are all zero pairs $(0, 0)$; \vspace{0.1cm}}
		\State \parbox[t]{8cm}{The indexes in $D_{\rho,d_j}$ are the row positions where the pairs $(i,j)$ with indexes $i \in \mathcal{D} \setminus D_{\rho,d_j}$ are assigned in the $(k+\nu)$-th column of $P_2, \ldots, P_r$; \vspace{0.1cm}}
		\EndWhile
		\State Return $P_1, \ldots, P_r$.
	\end{algorithmic}
\end{algorithm}

\begin{algorithm}
	\small
	\caption{$ValidPartitioning$
		\newline
		\textbf{Input}:$ValidPartitions, k, r, portion, run, step, J_{\nu}$;
		\newline
		\textbf{Output}: $\mathcal{D}_{d_j} = \{D_{1,d_j},\ldots,D_{r,d_j}\}$.}
	\label{Valid}	
	\begin{algorithmic}[1]
		\State Set $D = \{1,\ldots,\alpha\}$;
		\If{$run \ne 0$}
		\State Find $\mathcal{D}_{d_j}$ that satisfies Condition 1 and Condition 2;
		\Else
		\State Find $\mathcal{D}_{d_j}$ that satisfies Condition 2;
		\EndIf
		\State Return $\mathcal{D}_{d_j}$.
	\end{algorithmic}
\end{algorithm}

We illustrate the importance of \emph{Condition 1} and \emph{Condition 2} by revisiting the example with the $(9, 6)$ for $\alpha=9$.

\emph{\textbf{Example 1. (cont.)}}
	We analyze three different partitions $D_1$, $D_2$, and $D_3$, given in Table \ref{Comparison3Partitions}, that present the subsets for both $J_1$ and $J_2$. The partition $D_1$ is a \emph{valid partition} since both \emph{Condition 1} and \emph{Condition 2} are satisfied. The partition $D_2$ complies only with \emph{Condition 2} since none of the subsets in $\{\{1,2,3\},\{4,5,6\},\{7,8,9\}\}$ for the nodes in $J_1$ is equal to the subsets $\{\{1,5,9\},\{2,6,7\},\{3,4,8\}\}$ for the nodes in $J_2$, but it is not obtained by using regular $run$ and $step$ values generated by Alg. 2. Finally, the partition $D_3$ does not comply neither to \emph{Condition 1} nor \emph{Condition 2} since there are no regular $run$ and $step$, and the same subset $\{7,8,9\}$ is present for the nodes from both $J_1$ and $J_2$, i.e. in $\{\{1,2,3\},\{4,5,6\},\{7,8,9\}\}$ and in $\{\{1,3,5\},\{2,4,6\},\{7,8,9\}\}$. 
	As a consequence the average repair bandwidth of the code produced with partition $D_1$ is $2.67$ which is the lowest (and equal to the bandwidth in Eq. \ref{optimal}) compared to the repair bandwidths for a $(9, 6)$ code produced with partitions $D_2$ and $D_3$.
	\begin{figure}
		\centering
		\includegraphics[width=3.5in]{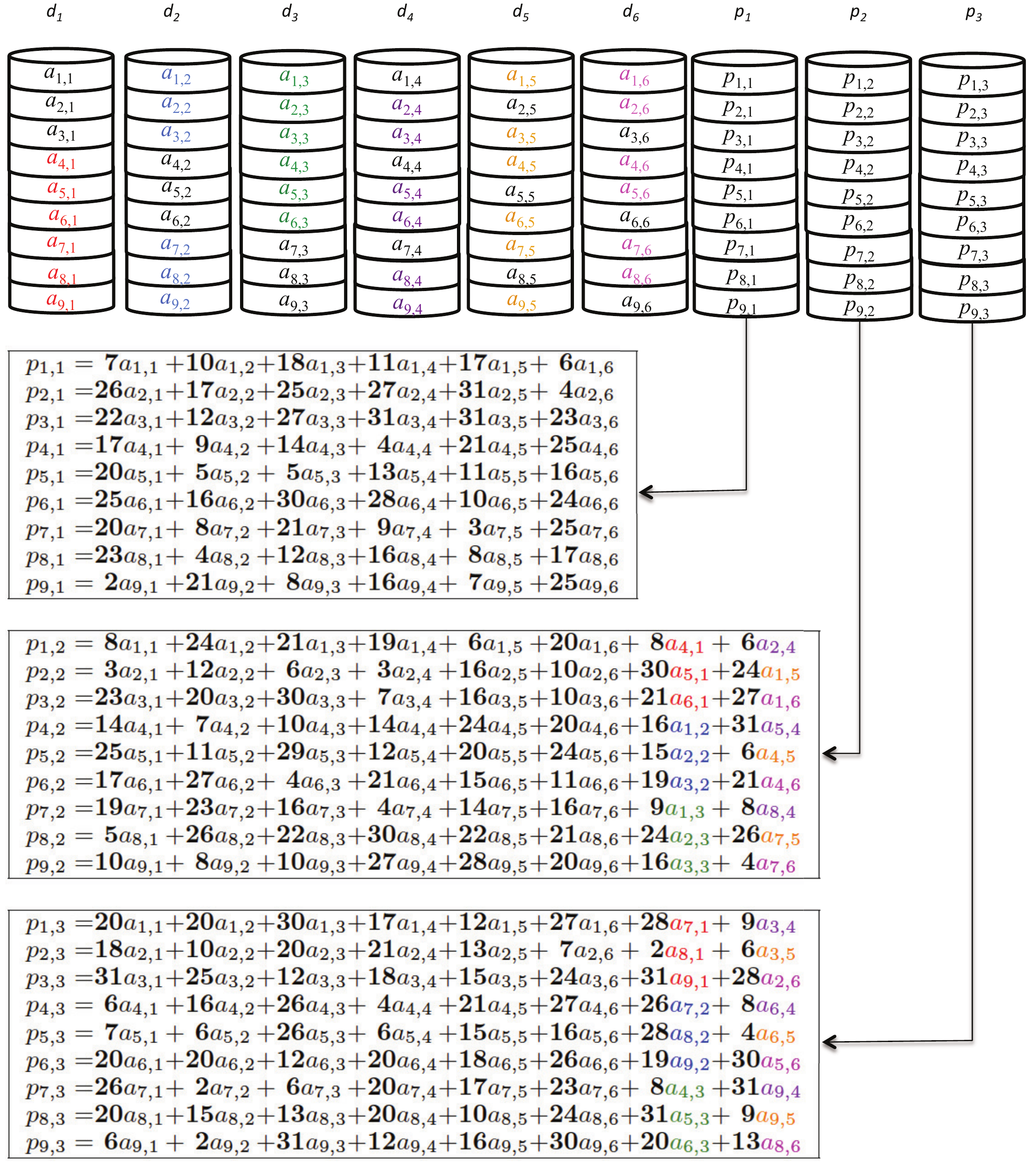}
		\vspace{-0.3cm}
		\caption{MDS array code with 6 systematic and 3 parity nodes for $\alpha=9$. The elements presented in colors are scheduled as additional elements in $p_2$ and $p_3$. The coefficients are from $\mathbf{F}_{32}$ with irreducible polynomial $x^5+x^3+1$.}
		\label{9_6_9}
		\vspace{-0.4cm}
	\end{figure}

\begin{table*}[t]
	\caption{Comparison of three partitions for a $(9, 6)$ code for $\alpha=9$ where the partitions for $J_1$ and $J_2$ are given.}\label{Comparison3Partitions}
	\vspace{-0.3cm}
	\centering
	\begin{tabular}{|c|c|c|c|c|}
		\hline	\rule{0pt}{3ex} 		
		Partitions & \parbox{2.4cm}{\vspace{0.1cm} \emph{Condition 1}\vspace{0.1cm} } & \emph{Condition 2} & \emph{ValidPartition} & \parbox{2.0cm}{\vspace{0.1cm} Avg. repair bw. \vspace{0.1cm} }\\
		\hline  \rule{0pt}{3ex}
		\parbox{8.1cm}{\vspace{0.1cm} \hspace{-0.3cm} $D_1 = \hspace{0.2cm} \{\{\mathcal{D}_{d_1} = \mathcal{D}_{d_2} = \mathcal{D}_{d_3}\}, \hspace{0.9cm} \{ \mathcal{D}_{d_4} = \mathcal{D}_{d_5} = \mathcal{D}_{d_6}\}\}= \\ =\{\{\{1,2,3\},\{4,5,6\},\{7,8,9\}\},\ \ \{\{1,4,7\},\{2,5,8\},\{3,6,9\}\}\}\\ \textcolor{white}{..........}\hspace{0.15cm} \downarrow \hspace{1cm} \downarrow \hspace{1cm} \downarrow \hspace{1.4cm} \downarrow \hspace{1cm} \downarrow \hspace{1cm} \downarrow\\
			\textcolor{white}{hahah}\hspace{0.01cm} D_{\rho,d_1} \hspace{0.4cm} D_{\rho,d_2}  \hspace{0.4cm}  D_{\rho,d_3}  \hspace{0.8cm} D_{\rho,d_4}  \hspace{0.4cm} D_{\rho,d_5}  \hspace{0.4cm} D_{\rho,d_6} \ \ \  $ 
		 \vspace{0.05cm} } & \parbox{2.4cm}{\vspace{0.1cm} \ \ \ \ \ \ \ \ \ Y \\ $run=3$, $step=0$ for $J_1$;\\$run=1$, $step=2$ for $J_2$.  \vspace{0.1cm} } & Y & Y & 2.67\\
		\hline \rule{0pt}{3ex} 
		\parbox{8.1cm}{\vspace{0.1cm} \hspace{-0.3cm} $D_2=\{\{\{1,2,3\},\{4,5,6\},\{7,8,9\}\},\ \ \{\{1,5,9\},\{2,6,7\},\{3,4,8\}\}\}$ 
	 \vspace{0.1cm} } & N & Y & N & 3.00\\
		\hline \rule{0pt}{3ex}
		\parbox{8.1cm}{\vspace{0.1cm} \hspace{-0.3cm} $D_3=\{\{\{1,2,3\},\{4,5,6\},\{7,8,9\}\},\ \ \{\{1,3,5\},\{2,4,6\},\{7,8,9\}\}\}$ 
			\vspace{0.1cm} } & N & N & N & 3.26\\
		\hline
	\end{tabular}
\end{table*}

Once the index arrays $P_1, \ldots, P_r$ are determined, the symbols $p_{i,l}$ in the parity nodes, $1 \leq i \leq \alpha$ and $1 \leq l \leq r$, are generated as a combination of the elements $a_{j_1,j_2}$ where the pair $(j_1, j_2)$ is in the $i$-th row of the index array $P_l$, i.e.
\begin{equation}\label{LinEquations}
p_{i,l}=\sum c_{l,i,j} a_{j_1,j_2}.
\end{equation}
The linear relations have to guarantee a MDS code, i.e. to guarantee that the entire information can be recovered from any $k$ (systematic or parity) nodes.

Kamath et al. in \cite{6846301} defined codes with locality as vector codes. Note that HTECs can be defined as vector codes as it is done in \cite{DBLP:journals/corr/GligoroskiKJS17}.

\subsection{MDS Property}
Next we show that there always exists a set of non-zero coefficients from $\mathbf{F}_q$ in the linear combinations given in Eq. (\ref{LinEquations}) so that a $(n, k)$ HTEC is MDS.
We adapt Theorem 4.1 from \cite{7084873} as follows:
\begin{theorem}\label{MDS}
	There exists a choice of non-zero coefficients $c_{l,i,j}$ where $l=1, \ldots, r,$ $i=1, \ldots, \alpha,$ and $j=1, \ldots, k$ from $\mathbf{F}_q$ such that the code is MDS if $q \geq \binom {n} {k} r \alpha$.
\end{theorem}

\begin{proof} 
	\ The system of linear equations given in Eq. (\ref{LinEquations}) defines a system of $r \times \alpha$ linear equations with $k \times \alpha$ variables.
	A repair of one failed node is given in Alg. 4 but for the sake of this proof, we explain the repair by discussing the solutions of the system of equations. When one node has failed, we have an overdetermined system of $r \times \alpha$ linear equations with $\alpha$ unknowns.
	In general this can lead to a situation where there is no solution.
	However, since the values in system (\ref{LinEquations}) are obtained from the values of the lost node, we know that there exists one solution. Thus, solving this system of $r \times \alpha$ linear equations with an overwhelming probability gives a unique solution, i.e. the lost node is recovered.
	When 2 nodes have failed, we have a system of $r \times \alpha$ linear equations with $2\alpha$ unknowns.
	The same discussion for the overdetermined system applies here.
	The most important case is when $r=n-k$ nodes have failed.
	In this case, we have a system of $r \times \alpha$ linear equations with $r \times \alpha$ unknowns.
	If the size $q$ of the finite field $\mathbf{F}_q$ is large enough, i.e. $q \geq \binom {n} {k} r \alpha$, as it is shown in Theorem 4.1 in \cite{7084873}, the system has a unique solution, i.e. the file $M$ can be collected from any $k$ nodes.
\end{proof}


From Theorem 1, it stands that HTECs as any other MDS codes are storage-reliability optimal meaning that they offer tolerance for $r$ arbitrary failures for the consumed storage.

\subsection{Repairing from a Single Systematic Failure}
From practitioner's point of view, the repair process first reads a set of $\lceil\sfrac{\alpha}{r}\rceil$ rows from the first parity node and the non-failed systematic nodes, and repairs only $\lceil\sfrac{\alpha}{r}\rceil$ elements from the failed systematic node. The essence of the algorithm is that the already read set is reused for repair of all subsequent elements. 
Alg. 4 shows how to repair a single systematic node where the systematic and the parity nodes are global variables.
A set of $\lceil \sfrac{\alpha}{r} \rceil$ symbols is accessed and transferred from each of $n-1$ helpers. If $\alpha \ne r^{\lceil \sfrac{k}{r} \rceil}$, then additional elements may be required as described in Step 4.
Note that a specific element is transferred just once and stored in a buffer. For every subsequent use of that element, the element is read from the buffer and a further transfer operation is not required.
The repair process is highly parallel because a set of $\lceil \sfrac{\alpha}{r} \rceil$ symbols is independently and in parallel repaired in Step 2, and then the remaining symbols are recovered in parallel in Step 5.
\vspace{-0.2cm}
\begin{algorithm}\label{repairS}
	\caption{Repair of a systematic node $d_l$
		\newline
		\textbf{Input:} $l$;
		\newline
		\textbf{Output:} $d_l$.}
	\label{Repair}
	\begin{algorithmic}[1]
		\State Access and transfer $(k-1) \lceil \sfrac{\alpha}{r}\rceil$ elements $a_{i,j}$ from all $k-1$ non-failed systematic nodes and $\lceil \sfrac{\alpha}{r}\rceil$ elements $p_{i,1}$ from $p_1$, where $i \in D_{\rho,d_l}$;
		\State Repair $a_{i,l}$, where $i \in D_{\rho,d_l}$;
		\State Access and transfer $(r-1)\lceil \sfrac{\alpha}{r}\rceil$ elements $p_{i,j}$ from $p_2, \ldots, p_{r}$, where $i \in D_{\rho,d_l}$;
		\State Access and transfer from the systematic nodes the elements $a_{i,j}$ listed in the $i-$th row of the arrays $P_2, \ldots, P_r$, where $i \in D_{\rho,d_l}$, that have not been read in Step 1;
		\State Repair $a_{i,l}$, where $i \in \mathcal{D} \setminus D_{\rho,d_l}$.
	\end{algorithmic}
\end{algorithm}

\subsection{Repair Bandwidth for a Single Systematic Failure}
The bandwidth optimality of the HTEC construction is captured in the following Proposition.
\begin{proposition} \label{key}
	If $r$ divides $\alpha$, then the indexes $(i, j)$ of the elements $a_{i,j}$, where $i \in \mathcal{D} \setminus D_{\rho,d_j}$, for each group of $r$ systematic nodes are scheduled in one of the $\lceil \sfrac{k}{r}\rceil$ additional columns in the index arrays $P_2,\ldots,P_r$.
\end{proposition}
\begin{proof}
	\ The proof is a simple counting strategy of all indexes $(i, j)$ of the elements $a_{i,j}$, where $i \in \mathcal{D} \setminus D_{\rho,d_j}$. 
\end{proof}

\begin{proposition} \label{bw}
	The bandwidth for repair of a single systematic node is bounded between the following lower and upper bounds:
	\begin{equation}
	\label{LowerUpperBounds}
	\frac{(n-1)}{r} \leq \gamma \leq \frac{(n-1)}{r} + \frac{(r-1)}{\alpha} \Bigl\lceil \frac{\alpha}{r} \Bigr\rceil
	\Bigl\lceil \frac{k}{r} \Bigr\rceil.
	\end{equation}
\end{proposition}

\begin{proof}
	\ Note that we read in total $k \lceil \sfrac{\alpha}{r}\rceil$ elements in Step 1 of Alg. 4. Additionally, $(r-1)\lceil \sfrac{\alpha}{r}\rceil$ elements are read in Step 3. Assuming that we do not read more elements in Step 4 and every element has a size of $\frac{1}{\alpha}$, we determine the lower bound as $\frac{(n-1)}{r}$. This bound is the same as the one given in Eq. (\ref{optimal}).
	To derive the upper bound, we assume that we read all elements $a_{i,j}$ from the extra $\lceil \sfrac{k}{r}\rceil $ columns of the arrays $P_2, \ldots, P_r$ in Step 4. Thus, the upper bound is $\frac{(n-1)}{r} + \frac{(r-1)}{\alpha} \Bigl\lceil \frac{\alpha}{r} \Bigl\rceil
	\Bigl\lceil \frac{k}{r}\Bigr\rceil$.
\end{proof}

HTECs are optimal in terms of the storage-bandwidth tradeoff for $\alpha=r^{\lceil\sfrac{k}{r}\rceil}$. In this case, HTECs achieve the bound of the repair bandwidth given in Eq. (\ref{optimal}). In all other cases, HTECs are near-optimal in terms of the storage-bandwidth tradeoff. Although HTECs are near-optimal for $\alpha < r^{\lceil\sfrac{k}{r}\rceil}$, they still achieve the lowest repair bandwidth compared to other representative codes from the literature as it is shown in Section \ref{anal}.

\begin{proposition}
	The recovery bandwidth is equal for all systematic nodes when $\alpha=r^{\lceil \sfrac{k}{r} \rceil}$. 
\end{proposition}
\begin{proof}
	\ When $\alpha=r^{\lceil \sfrac{k}{r} \rceil}$, Alg. 2 produces index arrays $P_1, \ldots, P_r$ where the distribution of the indexes from all systematic nodes is completely symmetric. The distribution of indexes always starts with $run = \sfrac{\alpha}{r}$ and $step=0$, and it ends with $run =\sfrac{\alpha}{r^{\lceil{\sfrac{k}{r}\rceil}}}=1$ and $step=\sfrac{\alpha}{r}-1$. That symmetry reflects to the linear dependencies in Eq. (\ref{LinEquations}) for each of the parity elements which further implies that the recovery bandwidth is symmetrical, i.e. equal for all systematic nodes. In order to repair any systematic node, the same amount of symbols is accessed from all $n-1$ nodes.
\end{proof}
This is illustrated with the examples for repairing a systematic node with a $(9, 6)$ HTEC for $\alpha=6$ and $9$ in Section 3.2. The distribution of the indexes in the $(9, 6)$ HTEC for $\alpha=9$ is symmetric as shown in Section 3.3, and the repair bandwidth for all systematic nodes is the same (24MB). 
On the other hand, the repair bandwidth for the systematic nodes differs with the $(9, 6)$ HTEC for $\alpha=6$. Namely, the repair bandwidth is 27MB for $d_1, d_3, d_4$ and $d_6$, while it is 30MB for $d_2$ and $d_5$. The scheduling of the indexes and the repair process for the $(9, 6)$ for $\alpha=6$ is thoroughly explained in Section \ref{explanation}.

\subsection{Repairing from Multiple Systematic Failures}
The same ideas for single failure repair apply to repair from multiple failures but here a larger set of rows is read. In the worse case when the number of failed nodes is $r$, then the data from all non-failed nodes is read.
Alg. 5 shows how to find a minimal system of linear equations to repair from $t$ failures, where $1\leq t\leq r$, with minimal bandwidth. Data from all $n-t$ non-failed nodes is accessed and transferred. 
The sets $N$ and $T$ consist of the indexes of all systematic nodes and the failed systematic nodes, respectively.
Note that for $t=1$ the amount of accessed and transferred data is the same with both Alg. 4 and Alg. 5.

\begin{algorithm}\label{repairM}
	\caption{Repair of $t$ systematic nodes, where $1\leq t \leq r$
		\newline
		\textbf{Input:} $T=\{l_1,\ldots,l_t\}$, where $T\subset N$ and $|T|=t$;
		\newline
		\textbf{Output:} Data from all $d_l$, where $l \in T$.}
	\label{Repair}
	\begin{algorithmic}[1]
		\For{each $l\in T$}
		\State Select equations $p_{i,l}$, where $i \in D_{\rho,d_l}$, from the parity nodes $p_1,\ldots, p_{r}$;
		\EndFor
		\While {The set of selected equations do not involve all 
			\phantom . \phantom . \phantom . \phantom . \phantom . \phantom . \phantom . missing $t\times\alpha$ elements $a_{i,l}$, where $i=1,\ldots,\alpha$ \phantom . \phantom . \phantom . \phantom . \phantom . \phantom . \phantom . \phantom . and $l\in T$} 
		\State \phantom . \phantom . Select equation $p_{i,j}$, where $i \in \mathcal{D} \setminus \cup_{j=l_1}^{l_t}D_{\rho,d_j}$, that \phantom . \phantom . \phantom . \phantom . \phantom . \phantom . includes maximum number of new non-included \phantom . \phantom . \phantom . \phantom . \phantom . elements $a_{i,l}$;
		\EndWhile
		\State Access and transfer from the available systematic nodes and from the parity nodes all elements $a_{i,j}$ and $p_{i,j}$ listed in the selected equations;
		\State Solve the system for $t\times\alpha$ unknowns $a_{i,l}$, where $i=1,\ldots,\alpha$ and $l\in T$;
		\State Return the data for the missing $d_l$, where $l \in T$.
	\end{algorithmic}
\end{algorithm}

\subsection{Repair Bandwidth for Multiple Systematic Failures}

\begin{proposition} \label{bw}
	The bandwidth to repair $t$ systematic nodes is bounded between the following lower and upper bounds:
	\begin{equation}
	\label{LowerUpperBounds}
	\frac{t}{\alpha}\Bigl\lceil\frac{\alpha}{r} \Bigr\rceil(n-t) \leq \gamma \leq k \alpha.
	\end{equation}
\end{proposition}

\begin{proof}\ Note that if for all missing nodes $d_l$, where $l \in T$, it stands that the index sets $D_{\rho,d_l}$ are disjunctive, i.e. it stands that $D_{\rho,d_{l_1}} \cap D_{\rho,d_{l_2}}=\emptyset$ where $l_1, l_2 \in T$, then in Steps 1 -- 3 of Alg. 5 we will select all $t \times \alpha$ necessary equations to repair the $t$ missing nodes. In that case Alg. 5 selects the minimum number of linear equations, thus, the repair bandwidth reaches the lower bound. This means that in Step 8 we need to read in total $t(k-t)\lceil\sfrac{\alpha}{r}\rceil$ elements $a_{i,j}$ from $k-t$ systematic nodes and additionally to read $t \cdot r \lceil \sfrac{\alpha}{r}\rceil$ elements $p_{i,j}$ from $r$ parity nodes. Assuming that every element has a size of $\sfrac{1}{\alpha}$, we determine the lower bound as $\frac{t}{\alpha} ((k-t)\Bigl\lceil\frac{\alpha}{r} \Bigr\rceil + r \Bigl\lceil\frac{\alpha}{r}\Bigr\rceil) = \frac{t}{\alpha}\Bigl\lceil\frac{\alpha}{r} \Bigr\rceil(n-t)$. 
	
	Any additional selection of equations in the while loop in Steps 4 -- 6 increases the repair bandwidth and cannot exceed the upper bound that is simply the same amount of repair bandwidth as for RS codes, i.e. $k \alpha$.
\end{proof}


\section{Code Examples with Arbitrary Sub-packetization Levels and Multiple Failures} \label{examples}
In this Section, we give two examples for a $(9, 6)$ HTEC code for $\alpha=6$ and $\alpha=9$. 
The $(9, 6)$ code is included in the latest release of Hadoop.

\subsection{A $(9, 6)$ HTEC for $\alpha=6$} \label{explanation}
The following requirements have to be satisfied for the code to be an access-optimal MDS code that achieves the lower bound of the repair bandwidth for any systematic node:
\begin{itemize}
	\item $M=k \alpha = 36$ symbols,
	\item Repair a failed systematic node by accessing and transferring $\lceil \frac{\alpha}{r}\rceil=2$ symbols from the remaining $d=8$ nodes,
	\item Reconstruct the data from any 6 nodes.
\end{itemize}
The systematic nodes $d_1, \ldots, d_6$ and the parity nodes $p_1, p_2, p_3$ are shown in Fig. \ref{9_6_6} where each node stores $\alpha=6$ symbols. In Fig. \ref{9_6_6}, we also show the elements $p_{i,l}$ from the parity nodes that are linear combinations from the elements $a_{i,j}$ from the systematic nodes. Both the elements from the systematic nodes and the coefficients in Eq. (\ref{LinEquations}) are replaced with concrete values in Fig. \ref{9_6_6}. The elements of $p_1$ are linear combinations of the row elements from the systematic nodes multiplied with coefficients from $\mathbf{F_{16}}$. The elements of $p_2$ and $p_3$ are obtained by adding extra symbols to the row sum. 
We next show the scheduling of an element $a_{i,j}$ from a specific $d_j$, where $i \in \mathcal{D} \setminus D_{\rho,d_j}$, at $portion=2$ positions in the $i$-th row, $i\in D_{\rho,d_j}$, and the $(6+\nu)$-th column, $\nu=1, 2,$ of $P_2$ and $P_3$.
We follow the steps in Alg. \ref{Encode} and give a brief explanation:
\newline
1) Initialize $P_i$, $i=1,2,3,$ as index arrays $P_i=((i, j))_{6\times 6}$,
	$$
	P_i=\left[\scriptsize
	\begin{array}{c@{\hspace{0.5em}}c@{\hspace{0.5em}}c@{\hspace{0.5em}}c@{\hspace{0.5em}}c@{\hspace{0.5em}}c@{\hspace{0.5em}}c@{\hspace{0.5em}}c@{\hspace{0.5em}}c@{\hspace{0.5em}}c@{\hspace{0.5em}}c@{\hspace{0.5em}}c@{\hspace{0.5em}}c}
	(1,1) & (1,2) & (1,3) & (1,4) & (1,5) & (1,6)\\
	(2,1) & (2,2) & (2,3) & (2,4) & (2,5) & (2,6)\\
	(3,1) & (3,2) & (3,3) & (3,4) & (3,5) & (3,6)\\ 
	(4,1) & (4,2) & (4,3) & (4,4) & (4,5) & (4,6)\\
	(5,1) & (5,2) & (5,3) & (5,4) & (5,5) & (5,6)\\
	(6,1) & (6,2) & (6,3) & (6,4) & (6,5) & (6,6)\\
	\end{array}
	\right].
	$$
\newline	
2) Append $\lceil \sfrac{k}{r}\rceil=2$ columns to $P_2$ and $P_3$ initialized to $(0, 0)$, i.e. 
	${\scriptsize P_2=P_3=}$$$\left[\scriptsize
	\begin{array}{c@{\hspace{0.5em}}c@{\hspace{0.5em}}c@{\hspace{0.5em}}c@{\hspace{0.5em}}c@{\hspace{0.5em}}c@{\hspace{0.5em}}c@{\hspace{0.5em}}c@{\hspace{0.5em}}c@{\hspace{0.5em}}c@{\hspace{0.5em}}c@{\hspace{0.5em}}c@{\hspace{0.5em}}c}
	(1,1) & (1,2) & (1,3) & (1,4) & (1,5) & (1,6) & (0,0) & (0,0)\\
	(2,1) & (2,2) & (2,3) & (2,4) & (2,5) & (2,6) & (0,0) & (0,0)\\
	(3,1) & (3,2) & (3,3) & (3,4) & (3,5) & (3,6) & (0,0) & (0,0)\\ 
	(4,1) & (4,2) & (4,3) & (4,4) & (4,5) & (4,6) & (0,0) & (0,0)\\
	(5,1) & (5,2) & (5,3) & (5,4) & (5,5) & (5,6) & (0,0) & (0,0)\\
	(6,1) & (6,2) & (6,3) & (6,4) & (6,5) & (6,6) & (0,0) & (0,0)\\
	\end{array}
	\right].
	$$
3) Set $portion$ equal to 2 and $ValidPartitions$ to an empty set.
\newline
4) For the systematic nodes $d_1, d_2,$ and $d_3$ in $J_1$, $run=2$ and $step=0$.
 \newline
5) Alg. \ref{Valid} returns $\mathcal{D}_{d_1}=\{\{1, 2\}, \{3, 4\}, \{5, 6\}\}$.
	Following Step 14 in Alg. \ref{Encode}, the first 2 zero pairs in the $7$-th column of $P_2$ with $0$ distance between them are at the positions $(1,7)$ and $(2,7)$, thus, $D_{\rho,d_1} =\{1,2\}$.
	\newline
6) We schedule the indexes of the elements of $d_1$ with $i$ indexes that are not elements of $D_{\rho,d_1}$ (written in red color in Fig. \ref{9_6_6} and in the arrays $P_2$ and $P_3$) in the 1-st and 2-nd row and 7-th column of $P_2$ and $P_3$.
		Similarly, we perform the same steps for the nodes $d_2$ and $d_3$ resulting in $D_{\rho,d_2} =\{3,4\}$ and $D_{\rho,d_3} =\{5,6\}$, respectively.
	\newline
	Next we schedule the elements from $d_4, d_5$ and $d_6$.
\newline	
7) For the nodes $d_4, d_5,$ and $d_6$ in $J_2$, $run=1$ and $step=1$.
\newline
8) We perform the same steps as for the nodes in $J_1$. Here we only give the corresponding $D_{\rho,d_j}$, i.e. $D_{\rho,d_4}=\{1,3\}$, $D_{\rho,d_5}=\{2,5\}$, and $D_{\rho,d_6}=\{4,6\}$.
\newline
9) After replacing the $(0,0)$ pairs with specific $(i,j)$ pairs, the final index arrays are:
	$$
	P_1=\left[\scriptsize
	\begin{array}{c@{\hspace{0.5em}}c@{\hspace{0.5em}}c@{\hspace{0.5em}}c@{\hspace{0.5em}}c@{\hspace{0.5em}}c@{\hspace{0.5em}}c@{\hspace{0.5em}}c@{\hspace{0.5em}}c@{\hspace{0.5em}}c@{\hspace{0.5em}}c@{\hspace{0.5em}}c@{\hspace{0.5em}}c}
	(1,1) & (1,2) & (1,3) & (1,4) & (1,5) & (1,6)\\
	(2,1) & (2,2) & (2,3) & (2,4) & (2,5) & (2,6)\\
	(3,1) & (3,2) & (3,3) & (3,4) & (3,5) & (3,6)\\ 
	(4,1) & (4,2) & (4,3) & (4,4) & (4,5) & (4,6)\\
	(5,1) & (5,2) & (5,3) & (5,4) & (5,5) & (5,6)\\
	(6,1) & (6,2) & (6,3) & (6,4) & (6,5) & (6,6)\\
	\end{array}
	\right],
	$$
	$$
	P_2=\left[\scriptsize
	\begin{array}{c@{\hspace{0.5em}}c@{\hspace{0.5em}}c@{\hspace{0.5em}}c@{\hspace{0.5em}}c@{\hspace{0.5em}}c@{\hspace{0.5em}}c@{\hspace{0.5em}}c@{\hspace{0.5em}}c@{\hspace{0.5em}}c@{\hspace{0.5em}}c@{\hspace{0.5em}}c@{\hspace{0.5em}}c}
	(1,1) & (1,2) & (1,3) & (1,4) & (1,5) & (1,6) & \textcolor{red}{(3,1)} & (2,4)\\
	(2,1) & (2,2) & (2,3) & (2,4) & (2,5) & (2,6) & \textcolor{red}{(4,1)} & (1,5)\\
	(3,1) & (3,2) & (3,3) & (3,4) & (3,5) & (3,6) & (1,2) & (5,4)\\ 
	(4,1) & (4,2) & (4,3) & (4,4) & (4,5) & (4,6) & (2,2) & (1,6)\\
	(5,1) & (5,2) & (5,3) & (5,4) & (5,5) & (5,6) & (1,3) & (3,5)\\
	(6,1) & (6,2) & (6,3) & (6,4) & (6,5) & (6,6) & (2,3) & (3,6)\\
	\end{array}
	\right],
	$$
	and
	$$
	P_3=\left[\scriptsize
	\begin{array}{c@{\hspace{0.5em}}c@{\hspace{0.5em}}c@{\hspace{0.5em}}c@{\hspace{0.5em}}c@{\hspace{0.5em}}c@{\hspace{0.5em}}c@{\hspace{0.5em}}c@{\hspace{0.5em}}c@{\hspace{0.5em}}c@{\hspace{0.5em}}c@{\hspace{0.5em}}c@{\hspace{0.5em}}c}
	(1,1) & (1,2) & (1,3) & (1,4) & (1,5) & (1,6) & \textcolor{red}{(5,1)} & (4,4)\\
	(2,1) & (2,2) & (2,3) & (2,4) & (2,5) & (2,6) & \textcolor{red}{(6,1)} & (4,5)\\
	(3,1) & (3,2) & (3,3) & (3,4) & (3,5) & (3,6) & (5,2) & (6,4)\\ 
	(4,1) & (4,2) & (4,3) & (4,4) & (4,5) & (4,6) & (6,2) & (2,6)\\
	(5,1) & (5,2) & (5,3) & (5,4) & (5,5) & (5,6) & (3,3) & (6,5)\\
	(6,1) & (6,2) & (6,3) & (6,4) & (6,5) & (6,6) & (4,3) & (5,6)\\
	\end{array}
	\right].
	$$
10) Schedule the elements $a_{i,j}$ with $(i,j)$ indexes represented in the index arrays. The parity symbols are linear combinations from the elements in the same row in the array. The coefficients for the MDS $(9, 6)$ code for $\alpha=6$ in Fig. \ref{9_6_6} are from $\mathbf{F}_{16}$ with irreducible polynomial $x^4+x^3+1$.

We next show how to repair node $d_1$ following Alg. 4. All symbols denoted by red rectangles in Fig. \ref{9_6_6} are accessed and transfered for repair of node $d_1$. First, we repair the elements $a_{1,1}, a_{2,1}$ since $1, 2 \in D_{\rho,d_1}$.
Thus, we access and transfer $a_{1,j}$ and $a_{2,j}$, where $j=2,\ldots,6$, from all 5 non-failed systematic nodes and $p_{1,1}$, $p_{2,1}$ from $p_1$.
Since $a_{3,1}, a_{4,1}$ are added as extra elements in $p_2$, we need to access and transfer $p_{1,2}$ and $p_{2,2}$ from $p_2$.
Due to the optimal scheduling of the extra elements in the parity nodes, no further elements are required to recover $a_{3,1}, a_{4,1}$.
The last two elements $a_{5,1}, a_{6,1}$ are recovered by accessing and transferring $p_{1,3}, p_{2,3}, a_{4,4}$ and $a_{4,5}$. Extra two elements are read because the sub-packetization level is not equal to 9. 
The data from $d_1$ is recovered by accessing and transferring in total 18 elements from 8 helpers.
Exactly the same amount of data (18 symbols) is needed to repair $d_3,d_4$ or $d_6$, while 20 symbols are needed to repair $d_2$ and $d_5$.
Thus, the average repair bandwidth is equal to 3.11 symbols.
The presented code is not optimal in terms of the repair bandwidth, i.e. an access and a transfer of more than 2 symbols from each of the non-failed nodes are required. 
Note that a $(9, 6)$ code for $\alpha=9$ is an access-optimal code.


\begin{figure}
	\centering
	\includegraphics[width=3.5in]{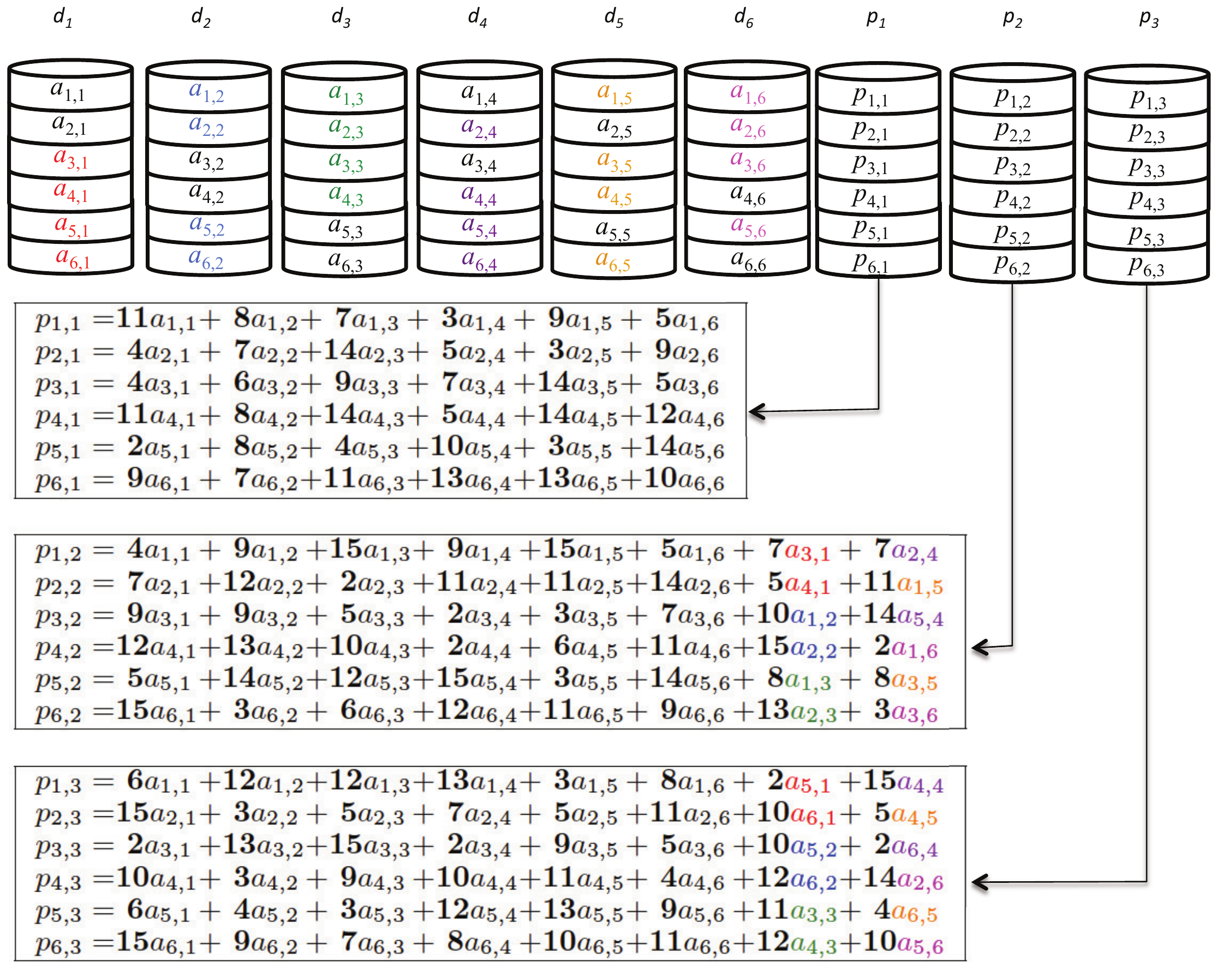}
	\vspace{-0.5cm}
	\caption{MDS array code with 6 systematic and 3 parity nodes for $\alpha=6$. The elements presented in colors are scheduled as additional elements in $p_2$ and $p_3$. The coefficients are from $\mathbf{F}_{16}$ with irreducible polynomial $x^4+x^3+1$.}
	\vspace{-0.4cm}
	\label{9_6_6}
\end{figure}


\subsection{A $(9, 6)$ HTEC for $\alpha=9$ and Repairing from Multiple Failures}
We next show the recovery of the nodes $d_1$ and $d_3$ with the $(9,6)$ code for $\alpha=9$ (Fig. \ref{9_6_9}) from Section \ref{algorithm}. 
We first access and transfer 24 $a_{i,j}$ elements from all 4 non-failed systematic nodes and 18 $p_{i,j}$ elements from $p_1,p_2,p_3$, where $i\in D_{\rho,d_1}\cup D_{\rho,d_3}=\{1,2,3, 7,8,9\}$. In total we have accessed and transferred 42 symbols. We next check if the number of linearly independent equations is equal to 18. When $d_1$ and $d_3$ are lost, this is fulfilled so it is possible to repair 18 lost symbols from $d_1$ and $d_3$. 
Exactly the same amount of data, 42 symbols, is needed to repair any pair of lost systematic nodes $d_{l_i}$ and $d_{l_j}$ for which it stands that $D_{\rho,d_{l_i}}\cap D_{\rho,d_{l_j}}=\emptyset$. There are in total $\binom{3}{2}$ combinations of 2 failed systematic nodes from the nodes $d_1, d_2,$ and $d_3$ in $J_1$ and $\binom{3}{2}$ combinations of 2 failed systematic nodes from the nodes $d_4, d_5,$ and $d_6$ in $J_2$. 

The recovery process of the nodes $d_{l_i}$ and $d_{l_j}$ has some additional steps when $D_{\rho,d_{l_i}}\cap D_{\rho,d_{l_j}}\neq\emptyset$. This happens when one node from each of the groups $J_1$ and $J_2$ has failed.
By reading the elements from Step 2 the number of linearly independent equations $p_{i,j}$ for $i \in D_{\rho,d_{l_i}}\cup D_{\rho,d_{l_j}}$ is exhausted. Thus, we have to read $p_{i,j}$ that has not been read previously, i.e. $p_{i,j}$ where $i \in \mathcal{D} \setminus  D_{\rho,d_{l_i}}\cup D_{\rho,d_{l_j}}$.
In order to illustrate this case, let us consider the repair of $d_1$ and $d_4$.
We first access and transfer 20 elements $a_{i,j}$ from all 4 non-failed systematic nodes and 15 $p_{i,j}$ elements from $p_1,p_2,p_3$, where $i\in D_{\rho,d_1}\cup D_{\rho,d_4}=\{1,2,3,4,7\}$. In total we have accessed and transferred 35 symbols. We next check if the number of linearly independent equations is equal to 18. Since we have only transferred 15 $p_{i,j},$ the condition is not fulfilled.
So we need to read 3 more $p_{i,j}$ that have not been read previously. In this case, we transfer $p_{5,1}, p_{5,2}, p_{6,1}$ and $a_{i,j}$ elements from the 5-th row in the parity arrays $P_1$ and $P_2$ and from the 6-th row in the parity array $P_1$ that have not been transferred in Step 2. The total number of symbols read to repair $d_1$ and $d_4$ is 46. There are in total $\binom{3}{1}\binom{3}{1}$ pairs of failed nodes where 46 symbols are needed to repair from double failures.

The average repair bandwidth to repair any 2 failed nodes is 4.933 symbols that is 17.783\% reduction compared to a $(9, 6)$ RS code.

\section{Optimizing I/O During Repair}\label{ioOpt}

Minimizing the amount of data accessed and transferred might not directly correspond to an optimized I/O unless the data reads are sequential. Motivated by the practical importance of I/O, we optimize HTECs while still retaining their optimality in terms of storage and repair bandwidth.

We first explain what we treat as a sequential and as a random read before discussing further sequential and random reads. Since the amount of data-read and the amount of data-transferred for HTECs are equal, the number of read and transfer operations is the same. Hence we use the terms reads and transfers interchangeably. Whenever there is a seek for data from a new location, the first read is counted as a random read. If the data is read in a contiguous manner, then the second read is counted as a sequential read. For instance, when a seek request is initiated for $a_{1,1}$ from $d_1$ in Fig. \ref{9_6_6_io}, then the number of random reads is 1. If we next read $a_{2,1}, a_{3,1},$ and so forth in a contiguous manner, then the number of sequential reads increases by one for each contiguous access. On the other hand, if we read $a_{3,1}$ after reading $a_{1,1}$ (but not $a_{2,1}$), then the number of random reads becomes 2. Note that reading $a_{6,1}$ and $a_{1,1}$ results in 2 random reads.

The parameter $step$ defines the contiguity of the reads for the codes obtained by Alg. \ref{Encode}.
\begin{proposition}
	The number of random reads for an $(n, k)$ access-optimal HTEC is equal to $n-1$ for $r$ out of $k$ systematic nodes. 
\end{proposition}
\begin{proof}
	\ When repairing a single systematic node with a $(n, k)$ HTEC for $\alpha=r^{\lceil \sfrac{k}{r} \rceil}$, then data from all $n-1$ helpers has to be accessed and transferred. The set of $k$ systematic nodes is partitioned in $\lceil\sfrac{k}{r}\rceil$ disjunctive subsets of $r$ nodes (the last subset may have less than $r$ nodes). For the first group of $r$ nodes in $J_1$, $step$ is equal to $0$, and hence the reads are sequential. There are in total $n-1$ seeks to read the data in a contiguous manner from $n-1$ helpers.
\end{proof}

\begin{figure}
	\centering
	\includegraphics[width=3.5in]{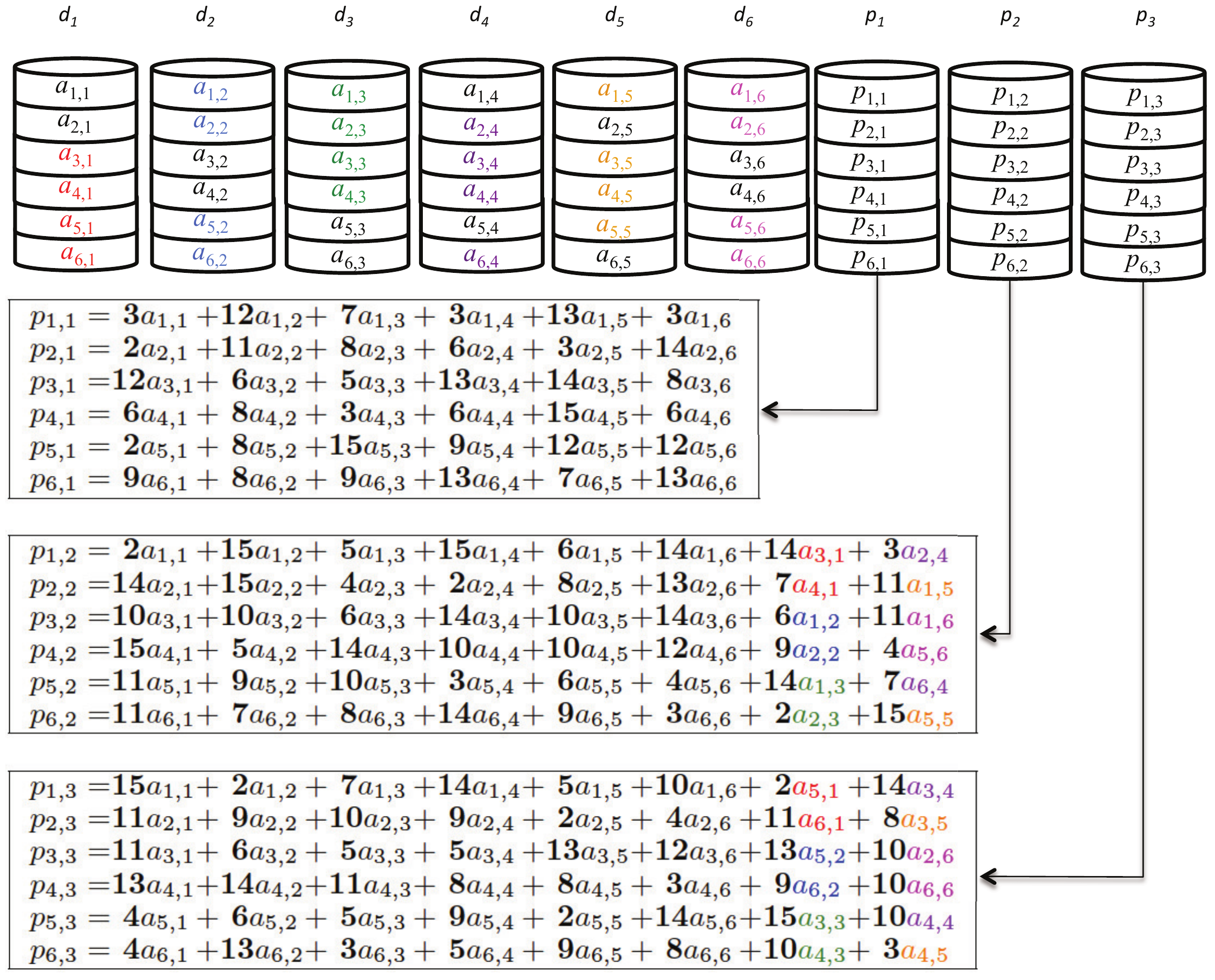}
	\vspace{-0.5cm}
	\caption{An I/O optimized  MDS  array code with 6 systematic and 3 parity nodes for $\alpha=6$. The elements presented in colors are scheduled as additional elements in $p_2$ and $p_3$. The coefficients are from $\mathbf{F}_{16}$ with irreducible polynomial $x^4+x^3+1$.}
	\vspace{-0.4cm}
	\label{9_6_6_io}
\end{figure}
Let us consider that the file size is 54MB and each node stores 9MB. Each I/O reads and transfers 512KB. When repairing a failed systematic node with a $(9,6)$ RS code, 6 out of 8 non-failed nodes have to be accessed. There are in total 6 random reads to recover 1 node. 
Since each I/O transfers 512KB and with the RS code the whole data of 9MB stored in a node is read and transferred, then there are 18 I/O transfers of 512KB from each node where the first I/O is random and the next 17 I/Os are sequential. Thus, the number of I/Os of 512KB is 108 where 6 are random and 6$\times$17=102 are sequential.

We next revisit the example from Section \ref{explanation} where $\alpha=6$. 
In that example, the data in each node is divided into blocks of 9MB$/$6=1.5MB. 
The average number of random reads for recovery of one systematic node is 13.33. Since the block size is 1.5MB and each random I/O transfers 512KB, then each random read is accompanied with 2 sequential reads. In addition, there are 5.33 sequential reads of blocks resulting into 5.33$\times$3 reads of 512KB.
In average there are 13.33 random I/Os and 42.66 sequential I/Os when reconstructing a lost systematic node, i.e., there are in total 56 I/Os.


The scheduling of the indexes by Alg. 2 ensures a gradual increase in the number of random reads, therefore there is no need for additional algorithms such as hop-and-couple \cite{188424} to make the reads sequential. If we want to further optimize the code in terms of I/O, then we can apply the hill climbing technique presented in Alg. 6.

\begin{algorithm}
	\caption{I/O optimization of a $(n,k)$ code  
		\newline
		\textbf{Input}: A $(n, k)$ code generated with Alg. \ref{HighLevel};
		\newline
		\textbf{Output}: An $(n, k)$ I/O optimized code.}
	\label{i_o}
	\begin{algorithmic}[1]
		\State Find a $(n, k)$ MDS erasure code where the parity nodes are generated with Alg. \ref{HighLevel};
		\State Repeatedly improve the solution by searching for codes with low I/O for single systematic node failure with same repair bandwidth, until no more improvements are necessary/possible.
	\end{algorithmic}
\end{algorithm}

With the help of Alg. 6, we reduce the number of random reads while still providing the same average repair bandwidth.
The code given in Fig. \ref{9_6_6_io} is a good example of an $(9,6)$ I/O optimized code for $\alpha=6$.
For this code construction the number of random I/Os is reduced to 11.33, while the number of sequential I/Os becomes 44.66.
This is achieved by using different $\mathcal{D}_{d_j}$ for the nodes $d_4, d_5,$ and $d_6$ from $J_2$. Namely, we obtained the following sets $D_{\rho,d_4}=\{1,5\}$, $D_{\rho,d_5}=\{2,6\}$, and $D_{\rho,d_6}=\{3,4\}$.
For instance, the ratio between the number of random I/Os and total number of I/Os for the non-optimized version of the code is 0.238, while it is 0.202 for the optimized version.
This is an important improvement since in practice random reads are more expensive compared to sequential reads.

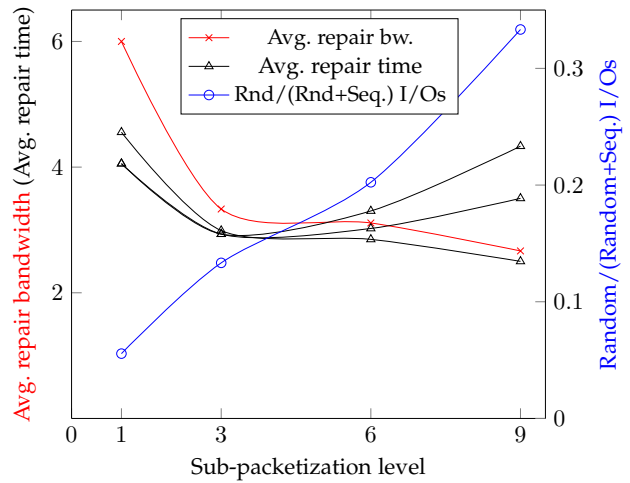
\begin{figure}
	\centering
	\begin{tikzpicture}[scale = 0.9]
	\pgfplotsset{
		scale only axis,
		xmin=0, xmax=9.5,
		xtick={0,1,3,6,9},
	}
	
	\begin{axis}[legend style={at={(0.23,0.98)},anchor=north west},
	axis y line*=left,
	ymin=0, ymax=6.5,
	ytick={2,4,6},
	xlabel=Sub-packetization level,
	ylabel={\color{red} Avg. repair bandwidth} (Avg. repair time),
	]
	\addplot[smooth,mark=x,red]
	coordinates{
		(1,6)
		(3,3.333)
		(6,3.111)
		(9,2.666)
	}; \label{plot_one}
	\addplot[smooth,mark=triangle,black]
	coordinates{
		(1,4.0555556)
		(3,2.93333)
		(6,3.302375)
		(9,4.333333)
	}; 
	\addplot[smooth,mark=triangle,black]
	coordinates{
		(1,4.0555556)
		(3,2.93333)
		(6,3.02375)
		(9,3.50333)
	};	
	\addplot[smooth,mark=triangle,black]
	coordinates{
		(1,4.555556)
		(3,2.99)
		(6,2.85)
		(9,2.5)
	};	\label{plot_two}
	\end{axis}
	
	\begin{axis}[legend style={at={(0.23,0.98)},anchor=north west},
	axis y line*=right,
	axis x line=none,
	ymin=0, ymax=0.35,
	ylabel={\color{blue} Random/(Random+Seq.) I/Os},
	] 
	\addlegendimage{/pgfplots/refstyle=plot_one}\addlegendentry{\small Avg. repair bw.}
	\addlegendimage{/pgfplots/refstyle=plot_two}\addlegendentry{Avg. repair time}
	\addplot[smooth,mark=o,blue]
	coordinates{
		(1,0.0555556)
		(3,0.133333)
		(6,0.202375)
		(9,0.333333)
	}; \addlegendentry{\small Rnd/(Rnd+Seq.) I/Os}
	\end{axis}
	
	\end{tikzpicture}
	\vspace{-0.3cm}
	\caption{Average repair bandwidth and normalized number of random I/Os for recovery of the systematic nodes for a $(9,6)$ code for different sub-packetization levels. The black U-shape parabolic curves are the expected curves for the average repair time for one systematic failure.}
	\label{9_6_io_parabolic}
	\vspace{-0.3cm}
\end{figure}

Fig. \ref{9_6_io_parabolic} shows the relation between the average repair bandwidth (we consider that $\frac{M}{k}=1$), the average repair time, and the normalized number of random reads for a single systematic failure with the sub-packetization level for a (9, 6) code.
We observe that for RS code where $\alpha=1$, the average repair bandwidth is biggest (the highest point on the red line with a value equal to 6), while the randomness in the I/Os is the lowest.
The repair bandwidth decreases as $\alpha$ increases and the minimum bandwidth of 2.67 is achieved for $\alpha=9$.
The situation is completely opposite when the metric of interest is the number of random reads. The number of reads (especially random reads) increases rapidly with $\alpha$ as shown in Fig. \ref{9_6_io_parabolic}. The best overall system performance is achieved for $\alpha$ in the range between 3 and 6.
This claim can be further clarified by the following reasoning: In terms of the total average time for a successful repair of one lost node, a higher average repair bandwidth means a higher average repair time. Similarly, a higher number of I/Os means a higher average repair time. Due to the opposite growth and drop trends of the curves for the average repair bandwidth and the normalized number of reads, we should expect U-shape parabolic curves for the total average time for a recovery of one node as the curves presented in black in Fig. \ref{9_6_io_parabolic}. In practical implementations, the concrete shape of the U-curve depends on additional parameters such as the speed of the disks, the amount and the speed of local disk cashes, the read latency, and the size of stored data blocks as we show in Section \ref{anal}.

\section{Experiments with HTECs in Hadoop}\label{anal}
To verify the performance of HTECs we implemented them in C/C++ and used them in Hadoop Distributed File System (HDFS). 
Hadoop is an open-source software framework used for distributed storage and processing of big data sets \cite{white2012hadoop}. From release 3.0.0-alpha2 Hadoop offers several erasure codes such as $(9, 6)$ and $(14, 10)$ RS codes.

All tests were run on publicly available Amazon EC2 instances running the default Ubuntu 64bit image. Hadoop 3.0.0-alpha2 was downloaded and installed on each machine. The erasure coding modules of HDFS were modified to support the HTEC C/C++ library. 
We used one namenode, nine data nodes, and one client node. All nodes had a size of 50GB and were connected with a local network of 10Gbps. The nodes were running on Linux machines equipped with Intel Xeon E5-2676 v3 running on 2.4GHz. Two crucial parameters in Hadoop are split size and block size (we refer an interested reader to \cite{white2012hadoop}). We have experimented with different block sizes (90MB and 360MB), different split sizes (512KB, 1MB and 4MB) and different sub-packetization levels ($\alpha=1, 3, 6,$ and $9$) in order to check how they affect the repair time of one lost node. The measured times to recover one node are presented in Fig. \ref{9_6_Hadoop01}. Note that $\alpha=1$ represents the RS code that is available in HDFS, while for $\alpha = 3, 6, 9$ the codes are HTECs defined in this paper. In all measurements HTECs outperform RS. The effect of the U-curves discussed in Fig. \ref{9_6_io_parabolic} is apparent for smaller split sizes, and as the split sizes increase, the disadvantage of bigger number of I/Os due to the increased sub-packetization diminishes, and the repair time decreases further. 
\begin{figure}[h!]
	\includegraphics[width=3.25in]{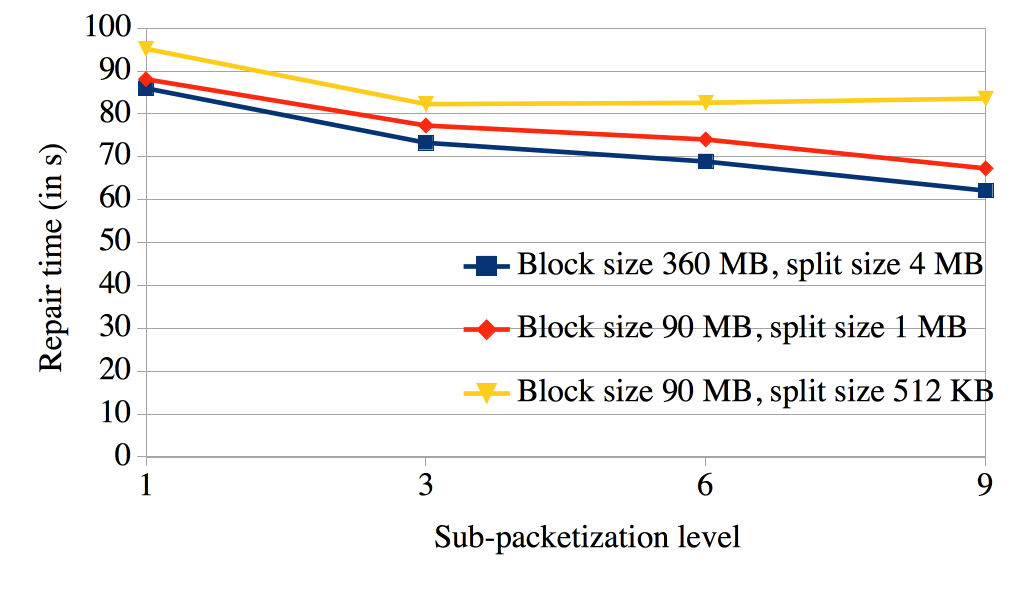}
	\vspace{-0.5cm}
	\caption{Time to repair one lost node of 50GB with a $(9, 6)$ code for different sub-packetization levels $\alpha$. Note that the RS code for $\alpha=1$ is available in the latest release 3.0.0-alpha2 of Apache Hadoop.}
	\vspace{-0.4cm}
	\label{9_6_Hadoop01}
\end{figure}

In Fig. \ref{9_6_Hadoop02}, we compare the repair times for one lost node of 50GB with codes that are directly obtained with Alg. 2, and I/O optimized codes obtained with Alg. 6. In almost all cases there is a small improvement (shorter repair time) with the I/O optimized codes.
\begin{figure}[h!]
	\includegraphics[width=3.3in]{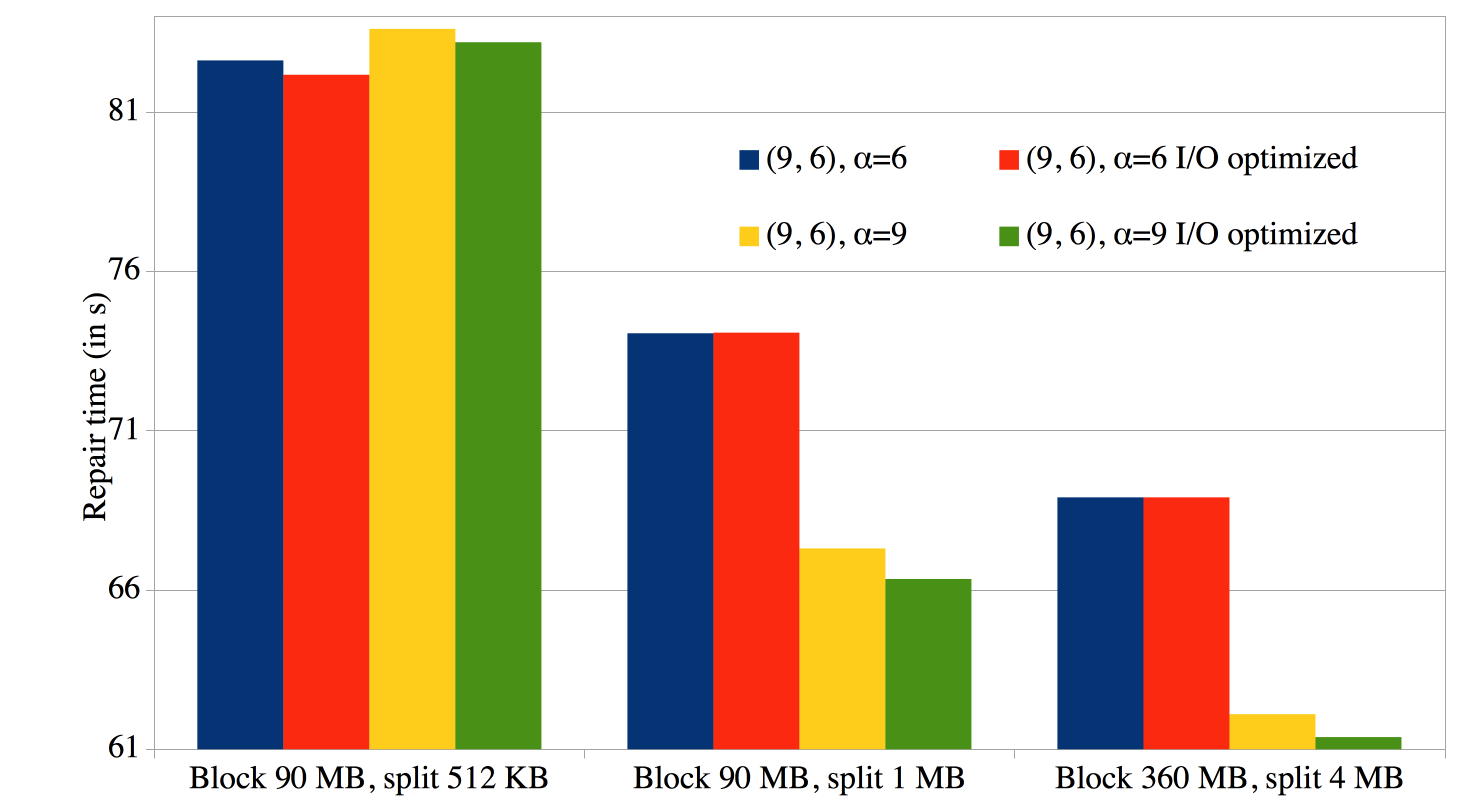}
	\vspace{-0.4cm}
	\caption{Comparison of repair times for one lost node of 50GB for a $(9, 6)$ code produced with Alg. 2 and an $(9, 6)$ I/O optimized code produced with Alg. 6 for sub-packetization levels equal to 6 and 9.}
	\label{9_6_Hadoop02}
	\vspace{-0.4cm}
\end{figure}

\subsection{Comparison of HashTag Codes with Other Codes} 
The performance of HTECs is further investigated in comparison with representative codes from the literature.
We first compare the average data that is both read and downloaded during a repair of a single systematic node for different code parameters with HTEC and Piggyback constructions \cite{6620242}. 
The plot in Fig. \ref{piggy} corresponds to a sub-packetization level equal to 8 in Piggyback 1 and HTEC, and $4(2r-3)$ in Piggyback 2.
We observe that HTEC construction requires less data read and less data transferred compared to Piggyback 1 and Piggyback 2 even though the sub-packetization level for the HTEC construction is smaller than the one in Piggyback 2.

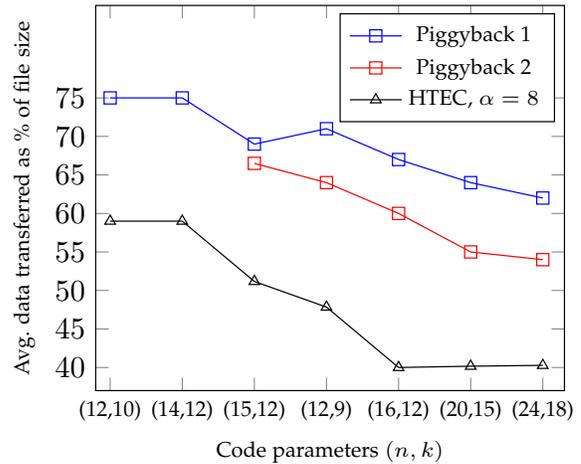
\begin{figure}
	\centering
	\begin{tikzpicture}[scale=1.15]
	\begin{axis}[
	xlabel=\text{\scriptsize Code parameters $(n, k)$},
	ylabel={\scriptsize Avg. data transferred as \% of file size},
	xmin=-0.2, xmax=6.3,
	ymin=37, ymax=87,
	xtick={0,1,2,3,4,5,6},
	ytick={40, 45, 50, 55, 60, 65, 70, 75},
	xticklabels={\scriptsize (12$\text{,}$10), \scriptsize (14$\text{,}$12), \scriptsize (15$\text{,}$12), \scriptsize (12$\text{,}$9), \scriptsize (16$\text{,}$12), \scriptsize (20$\text{,}$15), \scriptsize (24$\text{,}$18)}
	]
	\addplot[
	color=blue,
	mark=square,
	]
	coordinates {
		(0, 75)(1, 75)(2, 69)(3, 71)(4, 67)(5, 64)(6, 62)
	};
	
	\addplot[
	color=red,
	mark=square,
	]
	coordinates {
		(2, 66.5)(3, 64)(4, 60)(5, 55)(6, 54)
	};
	
	\addplot[
	color=black,
	mark=triangle,
	]
	coordinates {
		(0, 59)(1, 59)(2, 51.17)(3, 47.84)(4, 40)(5, 40.1667)(6, 40.27778)
	};
	\legend{\scriptsize Piggyback 1\\ \scriptsize Piggyback 2\\ \scriptsize HTEC, $\alpha=8$\\}
	
	\end{axis}
	\end{tikzpicture}
	\vspace{-0.5cm}
	\caption{Average data read and transferred for repair of a single systematic node with Piggyback 1 for $\alpha=8$, Piggyback 2 for $\alpha=4(2r-3)$ and HTEC for $\alpha=8$.}
	\label{piggy}
\end{figure}


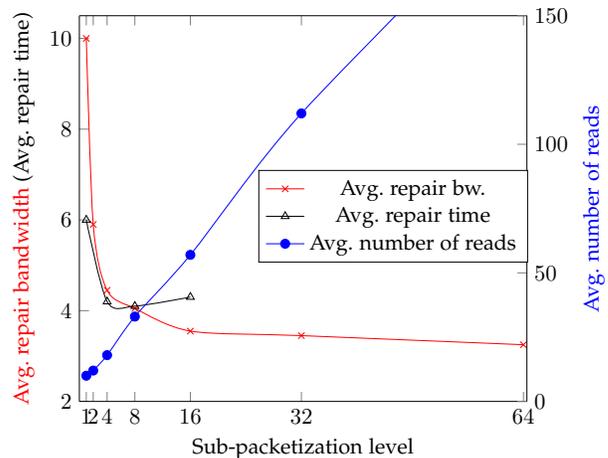
\begin{figure}
	\centering
	\begin{tikzpicture}[scale = 0.85]
	\pgfplotsset{
		scale only axis,
		xmin=0, xmax=64.5,
		xtick={1,2,4,8,16,32,64},
	}
	
	\begin{axis}[legend style={at={(0.4,0.6)},anchor=north west},
	axis y line*=left,
	ymin=2, ymax=10.5,
	ytick={2,4,6,8,10},
	xlabel=Sub-packetization level,
	ylabel={\color{red} Avg. repair bandwidth \color{black} (Avg. repair time)},
	]
	\addplot[smooth,mark=x,red]
	coordinates{
		(1, 10)
		(2, 5.9)
		(4, 4.45)
		(8, 4.05)
		(16, 3.55)
		(32, 3.45)
		(64, 3.25)
	}; \label{plot_one}
	\addplot[smooth,mark=triangle,black]
	coordinates{
		(1,6.0)
		(4,4.2)
		(8,4.1)
	   (16,4.3)
	}; \label{plot_two}
	\end{axis}
	
	\begin{axis}[legend style={at={(0.4,0.6)},anchor=north west},
	axis y line*=right,
	axis x line=none,
	ymin=0, ymax=150,
	ylabel={\color{blue} Avg. number of reads},
	] 
	\addlegendimage{/pgfplots/refstyle=plot_one}\addlegendentry{Avg. repair bw.}
	\addlegendimage{/pgfplots/refstyle=plot_two}\addlegendentry{Avg. repair time}
	\addplot[smooth,mark=*,blue]
	coordinates{
		(1, 10)
		(2, 12)
		(4, 18)
		(8, 33)
		(16, 57)
		(32, 112)
		(64, 200)
	}; \addlegendentry{Avg. number of reads}
	\end{axis}
	\end{tikzpicture}
	\vspace{-0.5cm}
	\caption{Average repair bandwidth for the systematic nodes and average number of reads (sequential+random reads) for recovery of the systematic nodes for a $(14,10)$ code for different sub-packetization levels. The black U-shape parabolic curve is the expected curve for the average repair time for one lost node.}
	\label{perf}
	\vspace{-0.3cm}
\end{figure}

Fig. \ref{perf} shows the relation between the average repair bandwidth (we consider that $\frac{M}{k}=1$) for a single failure, the average repair time, the average number of reads, and the sub-packetization level for a (14, 10) code.
For $\alpha=1$, we have a conventional RS code and the average repair bandwidth is equal to $k$ (the highest point on the red line with value 10).
A Hitchhiker code for $\alpha=2$ reduces the repair bandwidth by 35$\%$ compared to the RS code as it is reported in \cite{Rashmi:2014:HGF:2619239.2626325}, and the repair bandwidth is even further reduced by 41$\%$ with a $(14, 10)$ HTEC for $\alpha=2$. The remaining values of the average repair bandwidth are for HTECs for $\alpha$ = 4, 8, 16, 32, and 64. 
We observe that the lowest repair bandwidth that is 3.25 is achieved for $\alpha=r^{\lceil \sfrac{k}{r}\rceil}=64$. 
On the other hand, the highest number of reads is for $\alpha=64$. 
Typically, an engineering decision would end up choosing values for $\alpha$ in the range between 4 and 16 for optimal overall system performance in terms of the average repair time. 

That is illustrated in the next two figures: Fig. \ref{14_10_Hadoop01} and Fig. \ref{16_12_Hadoop01}. Fig. \ref{14_10_Hadoop01} presents the measured times to recover one node with a $(14, 10)$ code, and Fig. \ref{16_12_Hadoop01} presents the measured times to recover one node with a $(16, 12)$ code. The block sizes are between 128MB and 2048MB, and the split sizes are 1MB, 4MB, 64MB, and 128MB. 
The repair times are always better (i.e. lower) with HTECs compared to RS, and the repair times start to increase after a certain threshold for the sub-packetization level ($\alpha=8$ for the $(14, 10)$ code and $\alpha=16$ for the $(16, 12)$ code). This effect is more visible for small split sizes.
\begin{figure}[h!]
	\includegraphics[width=3.25in]{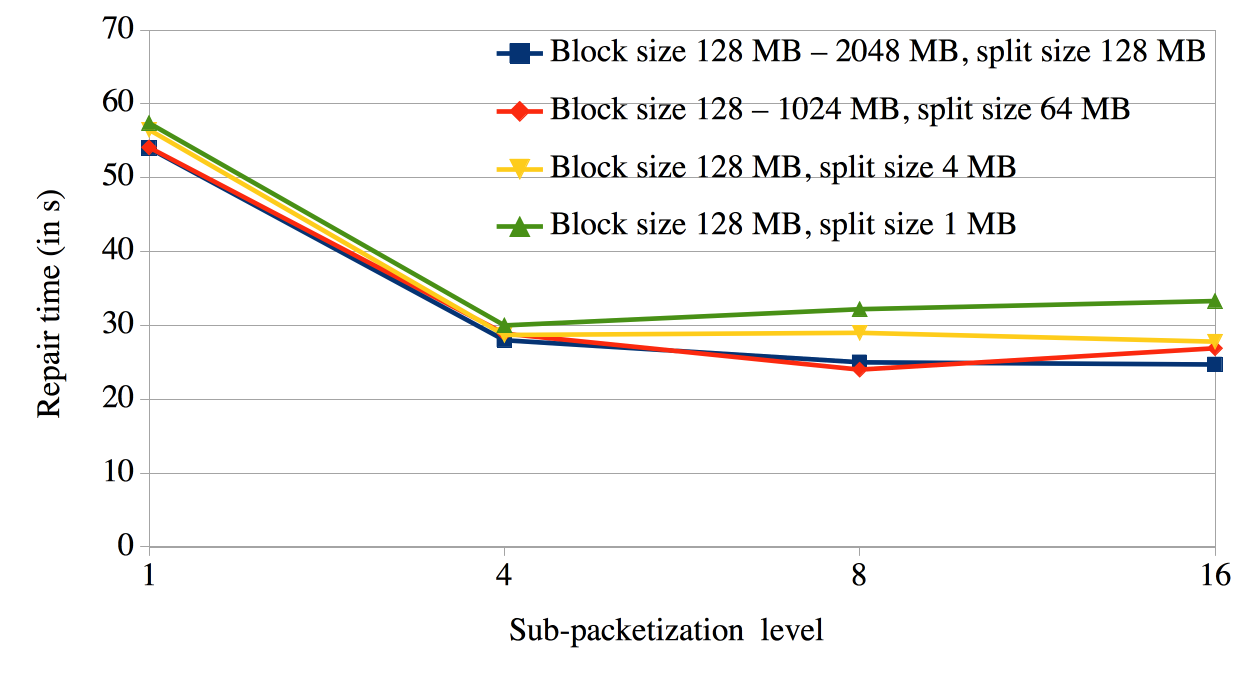}
	\vspace{-0.5cm}
	\caption{Time to repair one lost node of 50GB with a $(14,10)$ code for different sub-packetization levels $\alpha$.}
	\label{14_10_Hadoop01}
\end{figure}

\begin{figure}[h!]
\includegraphics[width=3.25in]{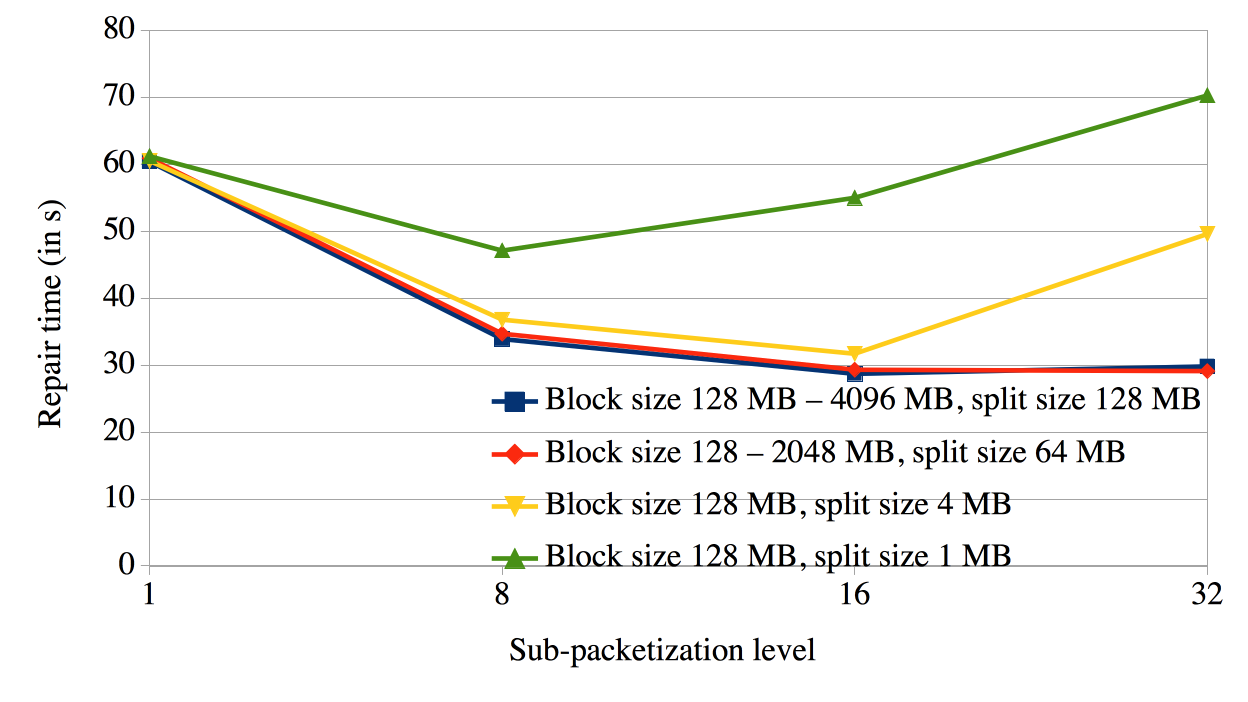}
\vspace{-0.5cm}
\caption{Time to repair one lost node of 50GB with a $(16,12)$ code for different sub-packetization levels $\alpha$. Note that the repair time increases for higher values of the sub-packetization level ($\alpha=16$ and $\alpha=32$) when the split sizes are small (green and yellow curve).}
\label{16_12_Hadoop01}
\end{figure}

\section{Discussion}\label{dis}
In this Section, we discuss some open issues that are not covered in this paper.

\emph{Lower bound of the finite field size.} In this paper, we use the work from \cite{7084873} to guarantee the existence of non-zero coefficients from $\mathbf{F}_{q}$ so that the code is MDS. However, the lower bound of the size of the finite field is relatively big. On the other hand, in all examples in this paper we actually work with very small finite fields ($\mathbf{F}_{16}$ and $\mathbf{F}_{32}$). 
Recent results in \cite{DBLP:journals/corr/YeB16a} showed that a code is access-optimal for $\alpha=r^{\lceil \sfrac{k}{r} \rceil}$ over any finite field $\mathbf{F}$ as long as $|\mathbf{F}| \geq r \Bigl\lceil \frac{k}{r} \Bigr\rceil$.
Determining the lower bound of the size of the finite field for HTECs remains an open problem.

\emph{Efficient repair of the parity nodes.} HTEC construction considers only an efficient repair of the systematic nodes. Several high-rate MSR codes for efficient repair of both systematic and parity nodes \cite{6120327,7282816,DBLP:journals/corr/YeB16} exist in the literature. Still for these codes, either the sub-packetization level is too large or the constructions are not explicit.
An open issue is how to extend the HTEC construction to support an efficient repair of the parity nodes as well.

\emph{Optimality in terms of I/O.} We use hill climbing technique in Alg. \ref{i_o} to optimize HTECs for the I/O. Finding HTECs that provably have the minimum I/O is an open optimization problem. 


\section{Conclusions}\label{conc}
MSR codes have been proposed as a superior alternative to popular RS codes in terms of minimizing the repair bandwidth. 
In this paper, we presented \emph{HashTag Erasure Code (HTEC)} construction that provides the flexibility of constructing MDS codes for any code parameters including an arbitrary sub-packetization level. MSR codes are constructed when the sub-packetization level of HTECs is equal to $r^{\lceil \sfrac{k}{r}\rceil}$. In this case, HTECs are access-optimal codes.

In this work we showed that when implemented in practical distributed storage systems such as in Hadoop, HTECs can provide the system designers great  flexibility in terms of selecting various code parameters such as the rate of the code, the size of the blocks and splits of the files, and the values of $\alpha$. Moreover, having in mind that the existing MDS erasure code constructions do not address the critical problem of I/O optimization, HTECs offer the possibility to further optimize the disk I/O consumed while simultaneously providing optimality in terms of storage, reliability, and repair-bandwidth. 
All these properties of HTECs offer the possibility to choose codes with parameters that give the best overall system performance.

Additionally, we show that HTECs reduce the repair bandwidth for more than one failure. HTECs are the first high-rate MDS codes theoretically constructed or implemented in practice, that offer significant improvements over RS codes in case of multiple failures.

\section*{Acknowledgements}\label{acknowledgemets}

We would like to thank Kjetil Babington for his practical insights. 

\bibliographystyle{IEEEtran}
\bibliography{refer}

\vspace{-0.4cm}
\begin{IEEEbiography}[{\includegraphics[width=1in,height=1.25in,clip,keepaspectratio]{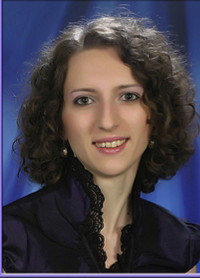}}]%
	{Katina Kralevska} is a postdoctoral researcher at the Department of Information Security and Communication Technology, NTNU. She was awarded a Ph.D. in December 2016 from NTNU. She received her B.Sc. degree in 2010 and her M.Sc. degree in 2012 in Telecommunications from Ss. Cyril and Methodius University-Skopje, Macedonia. Her research interests include applied erasure coding in networks and distributed storage systems.
\end{IEEEbiography}
\vspace{-0.4cm}
\begin{IEEEbiography}[{\includegraphics[width=1in,height=1.25in,clip,keepaspectratio]{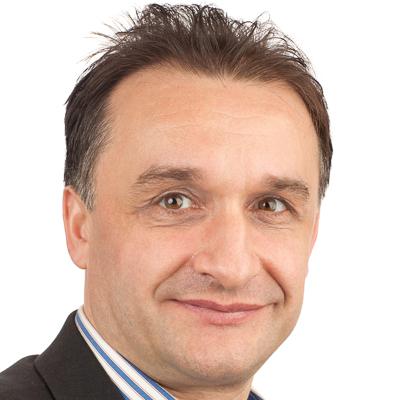}}]%
	{Danilo Gligorovski}
	 is a Professor of Information Security and Cryptography at NTNU. His main research interests are in Cryptography, Information security and Coding Theory, especially in Ultra Fast Public Key Algorithms, Post-Quantum Cryptography (multivariate and code-based), Hash functions, Fast Symmetric Cryptographic Algorithms and Erasure Codes for Distributed Storage Systems. He is an author of more than 170 scientific publications. 	
\end{IEEEbiography}
\vspace{-0.4cm}
\begin{IEEEbiography}[{\includegraphics[width=1in,height=1.25in,clip,keepaspectratio]{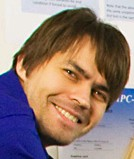}}]%
	{Rune E. Jensen}
	received his BSc degree in 2006 and his MSc degree in 2009 in Computer Science from Norwegian University of Science and Technology.
	He is currently working towards a PhD degree in optimization techniques for compute intensive applications at the Dept. of Computer and Information Science, Norwegian University of Science and Technology. 
	His interests are algorithmic and low level optimization in modern processors.
\end{IEEEbiography}
\vspace{-0.4cm}
\begin{IEEEbiography}[{\includegraphics[width=1in,height=1.25in,clip,keepaspectratio]{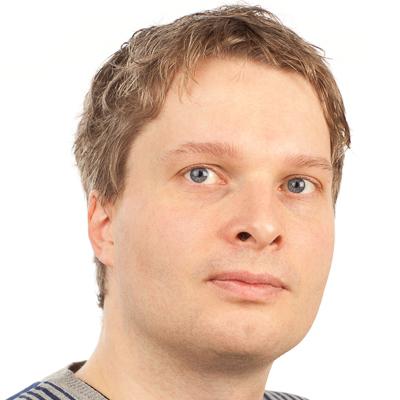}}]%
	{Harald {\O}verby}
	is a Professor at NTNU.
	He received his Msc in Computer Science in 2002, a BSc in Economics in 2003, and a PhD in Information and Communication Technology in 2005, all from NTNU. 
	He has held different administrative and academic positions at NTNU: Post.Doc (2005-2006), Research and Education Coordinator (2006-2010), and Associate Professor (2010-2016). His main research interests include digital economics, optical networking and secure and dependable communication systems.
\end{IEEEbiography}

\end{document}